\documentclass[runningheads,a4paper]{elsarticle}

\usepackage[english]{babel}
\usepackage[utf8]{inputenc}
\usepackage{amsthm}
\usepackage{amsmath,amssymb,amsfonts}
\usepackage{graphicx}
\usepackage[colorinlistoftodos]{todonotes}
\usepackage{xspace}
\usepackage{comment}
\usepackage{shuffle}

\usepackage{hyperref}

\usepackage{mathrsfs}

\usepackage{array}
\newcommand{\SV}[1]{}
\newcommand{\LV}[1]{#1}

\usepackage{fancybox}
\usepackage[mathscr]{euscript}
\usepackage{calligra}
\usepackage[titlenumbered,ruled]{algorithm2e}

\providecommand{\tabularnewline}{\\}

\theoremstyle{plain}
\newtheorem{theorem}{Theorem}
\theoremstyle{definition}
\newtheorem{definition}[theorem]{Definition}
\theoremstyle{plain}
\newtheorem{proposition}[theorem]{Proposition}
\theoremstyle{plain}
\newtheorem{lemma}[theorem]{Lemma}
\theoremstyle{plain}
\newtheorem{corollary}[theorem]{Corollary}
\theoremstyle{remark}
\newtheorem{example}[theorem]{Example}
\theoremstyle{remark}
\newtheorem{remark}[theorem]{Remark}

\DeclareMathAlphabet{\mathcalligra}{T1}{calligra}{m}{n}

\DeclareMathOperator{\exactbcov2}{\textsc{EBC}_2}
\DeclareMathOperator{\npclass}{\textsf{NP}}

\DeclareMathOperator{\pclass}{\textsf{P}}
\DeclareMathOperator{\nlclass}{\textsf{NL}}
\DeclareMathOperator{\pspaceclass}{\textsf{PSPACE}}
\DeclareMathOperator{\fa}{FA}
\DeclareMathOperator{\jfa}{JFA}
\DeclareMathOperator{\perm}{perm}

\DeclareMathOperator{\landau}{\mathcal{O}}

\definecolor{darkgreen}{rgb}{0,.7,0}
\definecolor{darkmagenta}{cmyk}{0,1,0,0.2}
\definecolor{HLB}{rgb}{.9,0,0}

\newcommand{\suchthat}{:}

\DeclareFontFamily{U}{bigshuffle}{}
\DeclareFontShape{U}{bigshuffle}{m}{n}{
  <5-8> s*[1.7] shuffle7
  <8->  s*[1.7] shuffle10
}{}
\DeclareSymbolFont{BigShuffle}{U}{bigshuffle}{m}{n}
\DeclareMathSymbol\bigshuffle{\mathop}{BigShuffle}{"001}
\DeclareMathSymbol\bigcshuffle{\mathop}{BigShuffle}{"002}

\usetikzlibrary{arrows}
\usetikzlibrary{arrows,automata}
\usetikzlibrary{automata,positioning,arrows}
\usetikzlibrary{shapes}
\usetikzlibrary{shapes,backgrounds}
\usetikzlibrary{matrix}
\usetikzlibrary{chains,fit,shapes}

\newcommand{\emptyword}{\varepsilon}
\newcommand{\shufflestar}{^{\shuffle,*}}
\newcommand{\shufflen}{^{\shuffle,n}}
\newcommand{\shufflenplusone}{^{\shuffle,n+1}}
\newcommand{\shufflei}{^{\shuffle,i}}
\newcommand{\shuffleiminusone}{^{\shuffle,i-1}}
\newcommand{\shufflezero}{^{\shuffle,0}}

\begin{document}
\title{Characterization and Complexity Results on Jumping Finite Automata}
  \author[add1]{Henning~Fernau\corref{cor1}}
  \ead{Fernau@uni-trier.de}
  \author[add1]{Meenakshi~Paramasivan}
  \ead{Paramasivan@uni-trier.de}
  \author[add1]{Markus~L.~Schmid}
  \ead{MSchmid@uni-trier.de}
   \author[add2]{Vojt\v{e}ch~Vorel}
  \ead{vorel@ktiml.mff.cuni.cz}

  \cortext[cor1]{Corresponding author}
  \address[add1]{Fachbereich 4 -- Abteilung Informatik, Universit\"at Trier, D-54286 Trier, Germany}
  \address[add2]{Faculty of Mathematics and Physics, Charles University, Malostransk\'{e} n\'{a}m. 25, Prague, Czech~Republic}

\begin{abstract}
In a jumping finite automaton, the input head can jump to an arbitrary position within the remaining input after reading and consuming a symbol. 
 We characterize the corresponding class of languages in terms of special shuffle expressions and survey other equivalent notions from the existing literature.
 Moreover, we present several results concerning computational hardness and algorithms for parsing and other basic tasks concerning jumping finite automata.
\end{abstract}

\begin{keyword}
Jumping Finite Automata\sep 
Shuffle Expressions\sep 
Commutative Languages\sep
Semilinear Sets\sep  
$\npclass$-hard problems\sep 
Exponential Time Hypothesis
\end{keyword}

\maketitle 


\section{Introduction}


Throughout the history of automata theory, the classical finite automaton has been extended in many different ways: two-way automata, multi-head automata, automata with additional resources (counters, stacks, etc.), and so on. However, for all these variants, it is always the case that the input is read in a continuous fashion. On the other hand, there exist models that are closer to the classical model in terms of computational resources, but that differ in how the input is processed (e.\,g., restarting automata \cite{Ott2006} and biautomata \cite{KliPol2012}). One such model that has drawn comparatively little attention are the jumping finite automata ($\jfa$) introduced by Meduna and Zemek~\cite{MedZem2012a,MedZem2014}, which are like classical finite automata with the only difference that after reading (i.\,e., consuming) a symbol and changing into a new state, 
the input head can jump to an arbitrary position of the remaining input. \par
We provide a characterization of the $\jfa$ languages in terms of expressions using shuffle, union, and 
iterated shuffle
, which enables us to put them into the context of classical formal language results.
Actually, we have discovered and became aware of many more connections of JFA languages with classical approaches to Formal Languages since we presented our
paper at CIAA in August 2015.
Hence,  this paper considerably deviates from and adds on to the conference version.

The contributions of this paper can be summarized as follows.
\begin{itemize}   
\item We introduce a variant of regular-like expressions, called alphabetic shuffle expressions, that characterize JFA languages and put them into the context of earlier literature
in the area of shuffle expressions. This also resolves several questions from \cite{MedZem2014}, as we show (this way) that JFA languages are closed under iterated shuffle. This approach also clarifies the closure properties under Boolean operations.\par
Moreover, we also investigate the intersection of the JFA languages with the regular languages and consider the problems of deciding for a given regular language or a given JFA language, whether or not it is also a JFA language or a regular language, respectively.
\item Alphabetic shuffle expressions are naturally related to semilinear sets, which allows us to derive a star-height one  normal form for these expressions.
\item We also discuss generalized variants of the two models presented so far, namely, \emph{general jumping finite automata} (\emph{GJFA}) and \emph{SHUF expressions}. In these models, transition labels (or axioms, respectively) are words, not single symbols only. We prove the incomparability of these two language classes and show how they relate to the class of JFA languages.
\item Finally, we arrive at several complexity results, in particular, regarding the parsing complexity of the various mechanisms. This is also one of the questions raised in \cite{MedZem2014}. 
For instance, it is demonstrated that there are fixed $\npclass$-complete GJFA languages. Furthermore, we strengthen the known hardness of the universal word problem for JFA by giving a lower bound based on the exponential time hypothesis.
\end{itemize}
\section{Preliminaries}

In this section, we present basic language-theoretical definitions and present the main concepts of this work, i.\,e., jumping finite automata and special types of shuffle expressions.

\subsection{Basic Definitions}\label{sec:basicDefinitions}

We assume the reader to be familiar with the standard terminology in formal language theory and language operations like catenation, union, and iterated catenation, i.\,e., Kleene star. For a word $w \in \Sigma^*$ and $a \in \Sigma$, by $|w|_a$ we denote the number of occurrences of symbol $a$ in $w$.\par
First, let us define the language operations of shuffle and permutation, and the
notion of semilinearity.

\begin{definition}
The \emph{shuffle operation}, denoted by $\shuffle$, is defined by 
\begin{eqnarray*}
u\shuffle v&=&\left\{ x_{1}y_{1}x_{2}y_{2}\ldots x_{n}y_{n}\suchthat \begin{array}{l}
u=x_{1}x_{2}\ldots x_{n},v=y_{1}y_{2}\ldots y_{n},\\
x_{i},y_{i}\in\Sigma^{*},1\leq i\leq n,n\geq1
\end{array}\right\}, \\
L_1\shuffle L_2&=&\bigcup_{\substack{x \in L_1\\ y \in L_2}} \left(u \shuffle v\right),
\end{eqnarray*}
for $u,v \in \Sigma^*$ and $L_1,L_2 \subseteq \Sigma^*$.
\label{def:1}
\end {definition}

\begin{remark}\label{rem-shuffle-inductive}
The following inductive definition of the shuffle operation,
equivalent to Definition~\ref{def:1},
shall 
be useful in some of our proofs. We have 
\begin{eqnarray*}
\emptyword \shuffle u &=& \{u\}, \\
u \shuffle \emptyword &=& \{u\}, \\
au\shuffle bv &=&  a(u \shuffle bv) \cup b(au \shuffle v),
\end{eqnarray*}
for every $u,v \in \Sigma^*$ and $a,b \in \Sigma$.
\end{remark}

The set of permutations of a word can be then conveniently defined using the shuffle operation.

\begin{definition}\label{def-perm}
The set $\perm(w)$ of all permutations of $w$ is inductively defined as follows: $\perm(\emptyword) = \{\emptyword\}$ and, for every $a \in \Sigma$ and $u \in \Sigma^*$, $\perm (a u) = \{a\} \shuffle \perm (u)$.
\end {definition}

The permutation operator extends to languages in the natural way. More precisely, 
$\perm(L_1) = \bigcup_{w \in L_1} \perm(w)$ for $L_1, L_2 \subseteq \Sigma^*$.
Analogously to the iterated catenation, an iterated shuffle operation can be defined as follows.

\begin{definition}
For $L \subseteq \Sigma^*$, the \emph{iterated shuffle} of $L$ is
$$
L\shufflestar = \bigcup^{\infty}_{n=0} L\shufflen ,
$$
where $L\shufflezero = \{\emptyword\}$ and $ L\shufflei = L\shuffleiminusone \shuffle L$.
\label{def:2}
\end {definition}

Let $\mathbb{N}$ denote the set of nonnegative integers and, for $n \geq 1$, let $\mathbb{N}^n$ be the $n$-fold Cartesian Product of $\mathbb{N}$ with itself. For $x,y\in \mathbb{N}^n$, i.\,e., $x=(x_1,\dots,x_n)$ and $y=(y_1,\dots,y_n)$, 
let $x + y = (x_1+y_1,\ldots,x_n+y_n)$ and for $c\in \mathbb{N}$, let $cx = (cx_1,\ldots,cx_n)$.

\begin{definition}
A subset $A \subseteq \mathbb{N}^n$ is said to be \emph{linear} if there are
$v, v_1, \ldots,v_m \in \mathbb{N}^n$ such that $$A = \{v + k_1v_1+k_2v_2+\cdots +k_mv_m \suchthat k_1,k_2,\ldots,k_m \in \mathbb{N}\}.$$ A subset $A \subseteq \mathbb{N}^n$ is said to be \emph{semilinear} if it is a finite union of linear sets.
\label{linearsemilinear}
\end {definition}

A permutation of the coordinates in $\mathbb{N}^n$ preserves semilinearity. Let $\Sigma$ be a finite set of $n$ elements. A \emph{Parikh mapping} $\psi$ from ${\Sigma}^*$ into $\mathbb{N}^n$ is a mapping defined by first choosing an enumeration $a_1, \ldots, a_n$ of the elements of $\Sigma$ and then defining inductively $\psi(\varepsilon) = (0,\ldots,0)$, $\psi(a_i) = (\delta_{1,i},\ldots,\delta_{n,i})$, where $\delta_{j,i} = 0$ if $i \neq j$ and $\delta_{j,i} = 1$ if $i = j$, and $\psi(au)=\psi(a)+\psi(u)$ for all $a\in\Sigma$, $u\in\Sigma^*$.
Any two Parikh mappings from ${\Sigma}^*$ into $\mathbb{N}^n$ differ only by a permutation of the coordinates of $\mathbb{N}^n$. Hence, the concept introduced in the following definition is well-defined.

\begin{definition}Let $\Sigma$ be a finite set of $n$ elements. 
A subset $A \subseteq {\Sigma}^*$ is said to be a \emph{language with the semilinear property}, or \emph{slip} language for short, if $\psi(A)$ is a semilinear subset of $\mathbb{N}^n$ for a Parikh mapping $\psi$ of ${\Sigma}^*$ into $\mathbb{N}^n$. The class of all slip languages is denoted by $\mathscr{PSL}$.
\label{sliplanguage}
\end {definition}

\subsection{Jumping Finite Automata and Shuffle Expressions}

Following Meduna and Zemek \cite{MedZem2012a,MedZem2014}, we denote a \emph{general finite machine} as $M = (Q, \Sigma, R, s, F)$, where $Q$ is a finite set of \emph{states}, $\Sigma$ is the \emph{input alphabet}, $\Sigma \cap Q = \emptyset$, $R$ is a finite set of \emph{rules}\footnote{We also refer to rules as \emph{transitions} with \emph{labels} from $\Sigma^*$.} of the form $py \rightarrow q$, where $p,q \in Q$ and $y \in \Sigma^*$, $s \in Q$ is the \emph{start state} and $F\subseteq Q$ is a set of \emph{final states}. If all rules $py \rightarrow q \in R$ satisfy $|y| \leq 1$, then $M$ is a \emph{finite machine}.\par
We interpret $M$ in two ways.
\begin{itemize}
\item As a (general) finite automaton: a \emph{configuration} of $M$ is any string in $Q\Sigma^*$, the binary \emph{move relation}
on $Q\Sigma^*$, written as $\Rightarrow$, is defined as follows:
$$pw\Rightarrow qz \iff \exists \ py \rightarrow q \in R \suchthat w = yz\,.$$
\item As a (general) jumping finite automaton: a \emph{configuration} of $M$ is any string in $\Sigma^*Q\Sigma^*$, the binary \emph{jumping relation} on $\Sigma^*Q\Sigma^*$, written as  $\curvearrowright$, satisfies: 
$$vpw\curvearrowright v'qz \iff \exists \ py \rightarrow q \in R \ \exists \ z\in \Sigma^* \suchthat w = yz \ \wedge \ vz = v'z'\,.$$
\end{itemize}
\noindent We hence obtain the following languages from a (general) finite machine $M$ :
\begin{align*}
L_{\fa}(M) &= \{w\in \Sigma^* \suchthat \exists \ f \in F \suchthat sw \Rightarrow^* f\},\\
L_{\jfa}(M) &= \{w\in \Sigma^* \suchthat \exists \ u,v \in \Sigma^* \ \exists \ f \in F \suchthat w=uv \wedge usv \curvearrowright^* f\}\,. 
\end{align*}
\noindent This defines the language classes $\mathscr{REG}$ (accepted by finite automata), $\mathscr{JFA}$ (accepted by jumping finite automata, or JFAs for short) and $\mathscr{GJFA}$ (accepted by general jumping finite automata, or GJFAs for short). As usual, $\mathscr{CFL}$ denotes the class of context-free languages.\par
Next, we define a special type of expressions that use the shuffle operator. Such shuffle expressions have been an active field of study over decades; we only point the reader to \cite{Jan81}, \cite{Jan85} and \cite{JedSze2001}. We first recall the definition of the \textsc{SHUF} expressions introduced by Jantzen~\cite{Jan79a},
from which we then derive $\alpha$-\textsc{SHUF} expressions, which are tightly linked to jumping finite automata.

\begin{definition}
The symbols $\emptyset, \emptyword$ and each $w \in \Sigma^+$ are
(atomic) \textsc{SHUF} expressions. If $S_1, S_2$ are \textsc{SHUF} expressions, then $(S_1+S_2), (S_1 \shuffle S_2)$ and ${S_1}\shufflestar$ are \textsc{SHUF} expressions.
\label{def:10}
\end {definition}

The semantics of \textsc{SHUF} expressions is defined in the expected way, i.\,e., $L(\emptyset) = \emptyset$, $L(\emptyword) = \{\emptyword\}$, $L(w) = \{w\}$, $w \in \Sigma^+$, and, for \textsc{SHUF} expressions $S_1$ and $S_2$, $L(S_1+S_2) = L(S_1) \cup L(S_2), L(S_1 \shuffle S_2) = L(S_1) \shuffle L(S_2)$, and $L({S_1}\shufflestar) = L({S_1})\shufflestar$.\par
A \textsc{SHUF} expression is an \emph{$\alpha$-SHUF expression}, if its atoms are only $\emptyset, \emptyword$ or single symbols $a \in \Sigma$. Since $\alpha$-\textsc{SHUF} expressions are \textsc{SHUF} expressions, the semantics are already defined.

Notice that we could introduce (classical) regular expressions in the very same way (i.\,e., we only have to substitute the shuffle operation by the catenation and the iterated shuffle by the Kleene star). Clearly, these characterize the regular languages.\par
Sometimes, to avoid confusion with arithmetics, we also write $\cup$ in expressions instead of $+$.\par
Let us illustrate the concepts defined above by two examples.

\begin{example}\label{ex:1}
Let $M$ be the finite machine depicted in Figure~\ref{fig:exampleFMM}, which accepts the regular language $L_{\fa}(M) = L((abc)^*)$. However, if we interpret $M$ as a jumping finite automaton, it accepts the non-context-free language $L_{\jfa}(M) = \{w \in \{a, b, c\}^* \suchthat |w|_a = |w|_b = |w|_c\}$. Obviously, $L_{\jfa}(M)$ is also defined by the $\alpha$-\textsc{SHUF} expressions $(a \shuffle b \shuffle c)\shufflestar$ and, furthermore, $\perm(L_{\fa}(M)) = L_{\jfa}(M)$. As shall be demonstrated later, every JFA-language can be expressed by an $\alpha$-\textsc{SHUF} expression and $\perm(L_{\fa}(M) = L_{\jfa}(M)$ holds for every finite machine $M$.
\end {example}
\begin{example}\label{ex:2}

The general finite machine $M'$ depicted in Figure~\ref{fig:exampleGFM} accepts the regular language $L_{\fa}(M') = L((abcd)^*)$. However, it is not easy to describe the language $L_{\jfa}(M')$ in a simple way.
Obviously, $\perm(L_{\fa}(M')) \neq L_{\jfa}(M')$ since $bacd \notin L_{\jfa}(M')$ and, furthermore, the \textsc{SHUF} expression $(ab \shuffle cd)^*$ does not describe $L_{\jfa}(M')$ either.
\end {example}


\begin{figure}[!htb]
\centering
\begin{minipage}{0.5\textwidth}
\centering
\begin{tikzpicture}[->,>=stealth',shorten >=1pt,auto,node distance=5.5em, thick]
  \tikzstyle{every state}=[fill=none,draw=black,text=black]
  \node[initial,state,accepting] (s)   [fill= none, draw=black,text=black]    {};
  \node[state]         (r)   [right of=s]                           {}; 
  \node[state]         (t)   [right of=r]                           {};
  \path (s) edge node [below] {$a$} (r)
   		(r) edge node [below] {$b$} (t)
        (t) edge [bend right] node [above] {$c$} (s);
\end{tikzpicture}
\caption{Finite Machine $M$.}
\label{fig:exampleFMM}
\end{minipage}%
\begin{minipage}{0.5\textwidth}
\centering
\begin{tikzpicture}[->,>=stealth',shorten >=1pt,auto,node distance=5.5em, thick]
  \tikzstyle{every state}=[fill=none,draw=black,text=black]
  \node[initial,state,accepting] (s)   [fill= none, draw=black,text=black]    {};
  \node[state]         (r)   [right of=s]                           {}; 
  \path (s) edge [bend left] node {$ab$} (r)
   		(r) edge [bend left] node {$cd$} (s);					
\end{tikzpicture}
\caption{General Finite Machine $M'$.}
\label{fig:exampleGFM}
\end{minipage}
\end{figure}


\subsection{Discussion of Related Concepts}

It is hard to trace back all origins and names of the concepts introduced so far.
We only mention a few of these sources in this subsection, also to facilitate finding the names of the concepts and understanding the connections to
other parts of mathematics and computer science. This subsection is not meant to be a survey on all the neighboring areas. Rather, it should give some impression about the richness of interrelations.
\begin{itemize}
\item Shuffle expressions have been introduced and studied to understand the semantics of parallel programs. This was undertaken, as it appears to be, independently by Campbell and
Habermann \cite{CamHab74}, by  Mazurkiewicz  \cite{Maz75} and by Shaw \cite{Sha78}. These expressions (also known as \emph{flow expressions}) allow for sequential operators (catenation and iterated catenation)
as well as for parallel operators (shuffle and iterated shuffle).
In view of the results obtained in this paper,
let us only highlight one particular result for these flow expressions. The universality problem was shown to be undecidable for such expressions, even when restricted to binary alphabets; see \cite{Iwa83}.
\item Further on, different variants of shuffle expressions in general were studied; see, e.\,g., \cite{FliKud2012b, Jan85, Jed90,  MaySto94}. We only mention  SHUF as a subclass of Shuffle Expressions. 
Actually, even the $\alpha$-SHUF expressions that we introduced in~\cite{FerParSch2015} have been
considered as a special case before; 
the corresponding class of languages was termed $\mathcal{L}_3$ in~\cite{HopOpp76}, possibly a bit
confusing given the traditional meaning of the term $\mathcal{L}_3$ for the regular languages. 
\item Semilinear subsets of $\mathbb{N}^n$ 
 show up quite naturally in many branches of mathematics.
 In our context of Theoretical Computer Science, the first important ingredient was the paper \emph{Semigroups, Presburger formulas, and languages} \cite{GinSpa66} whose title pretty much spells out the connections between logic and formal languages.\footnote{Interestingly enough, Presburger's original work
 \emph{Über die Vollständigkeit eines gewissen Systems der Arithmetik ganzer Zahlen, in welchem die Addition als einzige Operation hervortritt} has also a title that basically summarizes the contents of the paper.} 
  Later, connections to the theory of Petri nets, in particular, the famous reachability problem, and in this way also connections to many interesting properties of models for concurrent computations became a vivid research topic. 
We refrain from giving detailed further references here, as this would be surely  beyond this short historical summary; the interested reader can find many more references within the following papers:~\cite{DarDMM2012,EspNie94,FinLer2015,Rei2005}.
\item 
Eilenberg and Schützenberger started in \cite{EilSch69} the study of subsets (languages) of the free commutative monoid of some given alphabet, relating this again to earlier studies of Ginsburg and Spanier~\cite{GinSpa64} on bounded languages, which are kind of natural representatives of the permutation equivalence classes. 
With the notions given in \cite{EilSch69} for the definition of rational sets on the level of commutative monoids, if we replace the $+$ by $\shuffle$ and Kleene star by iterated shuffle, then we
basically arrive at the $\alpha$-SHUF expressions that we introduced in~\cite{FerParSch2015}.
The studies of \cite{EilSch69} were
later continued, e.\,g., by Huynh~\cite{Huy82,Huy83}.
\item Relations to blind counter automata become obvious if one considers
the way that JFAs process the input, basically only counting occurrences of different symbols. The connections to semilinear sets and to Petri nets were already discussed by Greibach~\cite{Gre78}; also, see \cite{Rei2005}.
Even more general yet related structures are studied in \cite{FerSti02,FerSti02a,MitSti2001,Zet2015}.
The main formal difference is that with JFA, the input is not processed continuously, while all devices
mentioned in this paragraph do process the input in a continuous manner, although the way that the (counter) storage can be operated allows these devices to incorporate some jump-like features.
\item Recently, K\v{r}ivka and Meduna~\cite{KriMed2015} studied \emph{jumping grammars} and also showed two variants of regular grammars that characterize $\mathscr{JFA}$ and $\mathscr{GJFA}$.
This type of grammar needs to be further compared to different varieties of 
commutative grammars. We only mention \cite{Esp97,Huy83,Nag0810,Nag2009}.
\end{itemize}

\section{Basic Algebraic Properties of Shuffle and Permutation}

In this section, we state some basic (algebraic) properties of the shuffle and permutation operations. To this end, we first recall the following computation rules for the shuffle operator from \cite{Jan79a}.

\begin{proposition} Let $M_1,M_2,M_3$ be arbitrary languages.
\begin{enumerate}
	\item $M_1 \shuffle M_2 = M_2 \shuffle M_1$ (commutative law),
	\item $(M_1 \shuffle M_2) \shuffle M_3 = M_1 \shuffle (M_2 \shuffle M_3)$ (associative law),
	\item $M_1 \shuffle (M_2 \cup M_3) = M_1 \shuffle M_2 \cup M_1 \shuffle M_3$ (distributive law),
	\item ${(M_1 \cup M_2)}\shufflestar = (M_1)\shufflestar \shuffle (M_2)\shufflestar$,
	\item ${({M_1}\shufflestar)}\shufflestar = (M_1)\shufflestar$,
	\item $(M_1\shuffle {{M_2}\shufflestar)}\shufflestar = (M_1 \shuffle{(M_1 \cup M_2)}\shufflestar)\cup \{\emptyword\}$.
\end{enumerate}
\label{laws:1}
\end{proposition}

The second, third and fifth rule are also true for (iterated) catenation  instead of (iterated) shuffle. This is no coincidence, as we will see. 
We can deduce from the first three computation rules the following result.

\begin{proposition}\SV{(*)}\label{prop-shuffle-semiring}
$(2^{\Sigma^*},\cup,\shuffle,\emptyset,\{\emptyword\})$ is a commutative semiring.
\end{proposition}

\LV{\begin{proof}
Instead of giving a complete formal argument, let us mostly recall what 
needs to be shown, giving then appropriate hints.
First, 
$(2^{\Sigma^*},\cup,\emptyset)$ is a commutative monoid;
this is a well-known set-theoretic statement.
Second, %
$(2^{\Sigma^*},\shuffle,\{\emptyword\})$ is a commutative monoid;
this corresponds to the first two computation rules, plus the fact that 
$\{\emptyword\}$ is the neutral element with respect to the shuffle operation.
Third, 
the distributive law was explicitly stated as the third computation rule.
\end{proof}

We are now discussing some special properties of the $\perm$ operator from a different (algebraic) viewpoint. 
 \LV{Reminiscent of the presentation in \cite{FerSem00}, there is an alternative way of looking at the permutation operator. Namely, let $w\in\Sigma^n$ be a word of length $n$, spelled out as
$w=a_1\cdots a_n$ for $a_i\in\Sigma$. Then, $u\in\perm(w)$ if and only if 
there exists a bijection $\pi:\{1,\dots,n\}\to\{1,\dots,n\}$ such that 
$u=a_{\pi(1)}\cdots a_{\pi(n)}$. In combinatorics, such bijections  
are also known as permutations. This also shows that 
$|\perm(w)|\leq (|w|)!$.
}

Next, we summarize two important properties of the operator $\perm$ in the following two lemmas.

\begin{lemma}\label{hullOperatorLemma}
$\perm:2^{\Sigma^*}\to 2^{\Sigma^*}$ is a hull operator, i.\,e., it is extensive ($L\subseteq \perm(L)$), increasing (if $L_1\subseteq L_2$, then $\perm(L_1)\subseteq \perm(L_2)$), and idempotent ($\perm(\perm(L))=\perm(L)$).
\end{lemma}


\LV{\begin{proof}We are going to show only the last of 
the three properties, the other two are easy to see.
Let $w\in \perm(\perm((L))\cap\Sigma^n$ with $w=a_1\cdots a_n$ for  $a_i\in\Sigma$.
This means that there is a permutation $\pi:\{1,\dots,n\}\to\{1,\dots,n\}$ such that $u=a_{\pi(1)}\cdots a_{\pi(n)}$ for some $u\in\perm(L)$.
This means that there is another permutation $\pi'$ such that 
$u'=a_{\pi'(\pi(1))}\cdots a_{\pi'(\pi(n))}\in L$.
As the composition of $\pi$ and $\pi'$ is again a permutation, we
find that $w\in \perm(L)$. Hence, $\perm(\perm(L))\subseteq \perm(L)$,
and as $\perm$ is extensive, $\perm(\perm(L))=\perm(L)$.
\end{proof}
}

\LV{Due to}\SV{By} the well-known correspondence between hull operators and (systems of) closed sets, we will
also speak \SV{of}\LV{about} \emph{perm-closed languages} in the following, i.\,e., languages $L$ satisfying $\perm(L)=L$. Such languages are also called 
\emph{commutative}\LV{, see \cite{Lat79a}}.
\begin{lemma}
The set $\left\{ \perm(w)\suchthat w\in \Sigma^* \right\}$
is a partition of $\Sigma^*$.
There is a natural bijection between this partition and
the set of functions $\mathbb{N}^\Sigma$, given by the \emph{Parikh mapping} $\pi_\Sigma:\Sigma^*\to \mathbb{N}^\Sigma, w\mapsto (a\mapsto |w|_a)$, where $|w|_a$ is the number of occurrences of $a$ in $w$.
Namely, $\perm(w)=\pi_\Sigma^{-1}(\pi_\Sigma(w))$ for $w\in \Sigma^*$.
%
\label{obs:1}
\label{permcharac}
\end {lemma}

Note that there exists a possibly better known semiring in formal language theory, using catenation instead of shuffle; let us make this explicit in the following statement.

\begin{proposition}\label{prop-catenation-semiring}
$(2^{\Sigma^*},\cup,\cdot ,\emptyset,\{\emptyword\})$
is a  semiring that is (in general) not commutative.
\end{proposition}}




Another algebraic interpretation can be given as follows:
\begin{proposition}
Parikh mappings can be interpreted as a semiring morphisms from
$(2^{\Sigma^*},\cup,\shuffle,\emptyset,\{\emptyword\}) $ to $(2^{\mathbb{N}^\Sigma},\cup,+,\emptyset,\{\vec{0}\} )$.
\end{proposition}

Especially, we conclude:

\begin{proposition}\label{prop-perm-Parikh}
For $u,v\in\Sigma^*$, $\perm(u)=\perm(v)$ if and only if $\pi_\Sigma(u)=\pi_\Sigma(v)$. 
For $L_1,L_2\subseteq\Sigma^*$, $\perm(L_1)=\perm(L_2)$
\LV{if and only if}\SV{iff}
$\pi_\Sigma(L_1)=\pi_\Sigma(L_2)$.
\end{proposition}

Due to Proposition~\ref{prop-perm-Parikh}, we can call
$u,v\in\Sigma^*$ (and also $L_1,L_2\subseteq \Sigma^*$)
\emph{per\-mu\-ta\-tion-equivalent} or 
\emph{Parikh-equivalent} if $\pi_\Sigma(u)=\pi_\Sigma(v)$
($\pi_\Sigma(L_1)=\pi_\Sigma(L_2)$, respectively).

The relation between (iterated) catenation and (iterated) shuffle can now be neatly expressed as follows.

\begin{theorem}\SV{(*)}\label{thm-semiring-morph}
$\perm:2^{\Sigma^*}\to 2^{\Sigma^*}$ is a semiring
morphism from the semiring
$(2^{\Sigma^*},\cup,\cdot,\emptyset,\{\emptyword\})$
to the semiring
$(2^{\Sigma^*},\cup,\shuffle,\emptyset,\{\emptyword\})$ that also respects the iterated catenation resp. iterated shuffle operation.
\end{theorem}

Clearly, $\perm$ cannot be an isomorphism, as the catenation semiring is not commutative, while the shuffle semiring is, see Proposition~\ref{prop-shuffle-semiring}. 

\LV{The proof of the previous theorem, broken into several statements that are also interesting in their own right, is presented in the following. Notice that in the terminology of \'Esik and Kuich~\cite{EsiKui2012},
Theorem~\ref{thm-semiring-morph} can also be stated as follows:
$\perm:2^{\Sigma^*}\to 2^{\Sigma^*}$ is a starsemiring
morphism from the starsemiring
$(2^{\Sigma^*},\cup,\cdot,{}^*,\emptyset,\{\emptyword\})$
to the starsemiring
$(2^{\Sigma^*},\cup,\shuffle,{}\shufflestar,\emptyset,\{\emptyword\})$.


\begin{lemma} $\forall u,v\in\Sigma^*$: 
$\perm(u \cdot v) = \perm(u) \shuffle \perm(v)$.
\label{lem:12}
\end{lemma}

\begin{proof}
We prove this lemma by induction on $|u|$. 

\noindent Induction Basis: $|u| = 1$. So, $u\in\Sigma$. By
Definition~\ref{def-perm}, 
$\perm(u \cdot v) = 
\{u
\} \shuffle \perm(v) = \perm(u) \shuffle \perm(v)$.

\noindent Induction Hypothesis: For $u\in\Sigma^n$, 
$\perm(u \cdot v) = \perm(u) \shuffle \perm(v)$.

\noindent Induction Step: Consider $|u| = n+1$. Let $u = x_1 x_2 \ldots x_{n+1}, \ x_i \in \Sigma^*$. We now claim that $\perm(x_1 x_2 \ldots x_{n+1} \cdot v) = \perm(x_1 x_2 \ldots x_{n+1}) \shuffle \perm(v)$.

\begin{align*}
\perm(x_1 x_2 \ldots x_{n+1} \cdot v) 	& =  \{x_1\} \shuffle \perm(x_2 \ldots x_{n+1} \cdot v) \ \text{(by Definition~\ref{def-perm})} \\
																			& =  \{x_1\} \shuffle \perm(x_2 \ldots x_{n+1}) \shuffle \perm(v) \ \text{(IH)} \\
																			& =  \perm(x_1 x_2 \ldots x_{n+1}) \shuffle \perm(v) \ \text{(by Definition~\ref{def-perm})}.			
\end{align*}

\noindent Therefore, $\perm(u \cdot v) = \perm(u) \shuffle \perm(v)$.
\end{proof}

\begin{lemma} $\forall u,v\in\Sigma^*$: 
$u \shuffle v \subseteq \perm(u \cdot v)$.
\label{lem:11}
\end{lemma}

\begin{proof}
By Definition~\ref{def:1}, $u \shuffle v = \{x_1y_1x_2y_2 \ldots x_n y_n \suchthat u = x_1 x_2 \ldots x_n, v = y_1 y_2 \ldots y_n, \linebreak[3] x_i, y_i \in \Sigma^*, 1 \leq i \leq n, \ n \geq 1 \}$. It is clear that $u \shuffle v \subseteq \perm(u) \shuffle \perm(v)$, as $\perm$ is a hull operator.
According to Lemma \ref{lem:12} $\perm(u) \shuffle \perm(v) = \perm(u \cdot v)$. Therefore, $u \shuffle v \subseteq \perm(u \cdot v)$. 
\end{proof}

\noindent
As a consequence of Lemma \ref{lem:11} and since $\perm$ is a hull operator, we obtain the following lemma.

\begin{lemma} $\forall u,v\in\Sigma^*$: 
$\perm(u \shuffle v) = \perm(u \cdot v)$.
\label{cor:2}
\end{lemma}

\begin{proof}
As indicated, from $u \shuffle v \subseteq \perm(u \cdot v)$
we can conclude that
$\perm(u \shuffle v)\subseteq \perm(\perm(u \cdot v))=\perm(u \cdot v)$.
Conversely, as $\{u\cdot v\}\subseteq u\shuffle v$,
$\perm(u\cdot v)\subseteq \perm (u\shuffle v)$.
\end{proof}

\noindent
This shows immediately, together with Lemma \ref{lem:12}:

\begin{lemma} $\forall u,v\in\Sigma^*$: 
$\perm(u \shuffle v) = \perm(u) \shuffle \perm(v)$.
\label{lem:14}
\end{lemma}

\begin{lemma}
$\perm(L^{n+1}) = \perm(L^n \shuffle L)$.
\label{lem:10}
\end{lemma}

\begin{proof}
The inclusion $\perm(L^{n+1}) \subseteq \perm(L^n \shuffle L)$ is true, since $L^{n+1} \subseteq L^n \shuffle L$. We now prove the other inclusion $\perm(L^{n+1}) \supseteq \perm(L^n \shuffle L)$. Let $w \in L^n \shuffle L$, then $\exists u \in L^n, v \in L \suchthat w \in u \shuffle v$. This implies that $\exists u \in L^n, v \in L \suchthat w \in \perm(u \cdot v)$ by Lemma \ref{lem:11}. Therefore, $\perm(L^{n+1}) = \perm(L^n \shuffle L)$. 
\end{proof}


\begin{lemma} 
Let $L,L_1,L_2\subseteq \Sigma^*$. Then
\begin{enumerate}
\item $\perm(L_1) \shuffle \perm(L_2) = \perm(L_1 \shuffle L_2) = \perm(L_1 \cdot L_2)$ and 
\item $(\perm (L))\shufflestar = \perm (L\shufflestar) = \perm (L^*)$.
\end{enumerate}
\label{thm:1}
\end{lemma}

\begin{proof} We are going to prove both parts separately.
\begin{enumerate}
\item Lemma~\ref{cor:2} immediately shows that the second equality holds.
Consider $L_1 \subseteq \Sigma^*$, $L_2 \subseteq \Sigma^*$. Let $w \in \perm(L_1) \shuffle \perm(L_2)$. Let $x' \in \perm(L_1)$, $y' \in \perm(L_2)$ such that $w \in x' \shuffle y'$. Hence,  there exists some $x \in L_1$ such that $x' \in \perm(x)$ (also $x \in \perm(x')$). Likewise, there exists some $y \in L_2$ with $y' \in \perm(y)$.
Hence, $w \in \perm(x) \shuffle \perm(y)=\perm(x \shuffle y)$ by Lemma~\ref{lem:14}. Therefore, $w \in \perm(L_1 \shuffle L_2)\cap \perm(L_1\cdot L_2)$. Similarly, if $w \in \perm(L_1 \shuffle L_2)$ then $w \in \perm(L_1) \shuffle \perm(L_2)$. Hence $\perm(L_1) \shuffle \perm(L_2) = \perm(L_1 \shuffle L_2)$.




\item

We will prove $(\perm (L))\shufflen = \perm(L\shufflen)$ and $\perm(L\shufflen) = \perm(L^{n})$ by induction on $n$. 


Induction Basis: $(\perm(L))\shufflezero = \{\emptyword\} = \perm(\emptyword) = \perm (L\shufflezero)$. 

Induction Hypothesis: $(\perm(L))\shufflen = \perm(L\shufflen)$. 

Induction Step: We now claim that $(\perm(L))\shufflenplusone = \perm(L\shufflenplusone)$.
\begin{align*}
(\perm(L))\shufflenplusone 	& =  (\perm(L))\shufflen \shuffle \perm(L) \ \text{(By Definition \ref{def:2})} \\
							& =  \perm(L\shufflen) \shuffle \perm(L) \ \text{(By Induction Hypothesis)} \\
							& = \perm(L\shufflen \shuffle L)  \ \text{(By (1) in Lemma \ref{thm:1}})\\
                            & = \perm(L\shufflenplusone) \ \text{(By Definition \ref{def:2})}.			
\end{align*}

We now prove $\perm(L\shuffleiminusone) = \perm(L^{i-1})$ by induction on $i$.

Induction Basis: $\perm(L\shufflezero) = \perm(\emptyword) = \{\emptyword\} = \perm(\emptyword) = \perm (L^0)$. 

Induction Hypothesis: $\perm(L\shufflen) = \perm(L^n)$. 

Induction Step: We now claim that $\perm(L\shufflenplusone) = \perm(L^{n+1})$.
\begin{align*}
\perm(L\shufflenplusone) 	& = \perm(L\shufflen \shuffle L) \ \text{(By Definition \ref{def:2})} \\
							& =  \perm(L\shufflen) \shuffle \perm(L) \ \text{(By (1) in Lemma \ref{thm:1})} \\
							& =  \perm(L^n) \shuffle \perm(L) \ \text{(By Induction Hypothesis)} \\
							& =  \perm(L^n \shuffle L) \ \text{(By (1) in Lemma \ref{thm:1})} \\
							& =  \perm(L^{n+1}) \ \text{(By Lemma \ref{lem:10})}.			
\end{align*}

By the very definitions of iterated catenation (Kleene star) and iterated shuffle, the claim of the second part follows.
\end{enumerate}
\end{proof}

\subsection*{Proof of Theorem~\ref{thm-semiring-morph}}

\begin{proof}
Recall that $\perm$, in order to be a semiring morphism, should satisfy the following properties:
\begin{itemize}
\item $\forall L_1,L_2\subseteq \Sigma^*: \perm(L_1\cup L_2)=\perm(L_1)\cup \perm(L_2)$.\\
{\small This is an easy standard set-theoretic argument.}
\item  $\forall L_1,L_2\subseteq \Sigma^*: \perm(L_1\cdot L_2)=\perm(L_1)\shuffle \perm(L_2)$.\\
{\small This was shown in Lemma~\ref{thm:1}.}
\item $\perm(\emptyset)=\emptyset$ and  $\perm(\{\emptyword\})=\{\emptyword\}$ are trivial claims.
\end{itemize}
Furthermore, we claim an according preservation property for the iterated catenation resp. shuffle, which is explicitly stated and proven in
Lemma~\ref{thm:1}.
\end{proof}}

\begin{remark}\label{rem-algebra}
Let us make some further algebraic consequences explicit.
\begin{itemize}

\item $\perm(L)$ can be seen as the canonical representative of all languages $\tilde L$ that are permutation-equivalent to $L$.
\item 
As $\perm$ is a morphism, there is in fact a semiring isomorphism between the permutation-closed languages (over $\Sigma$)
and $\mathbb{N}^{|\Sigma|}$, which is basically a Parikh mapping in this case.
\item There is a further natural isomorphism between the monoid
 $(\mathbb{N}^{|\Sigma|},+,\vec{0})$ and the free commutative monoid generated by $\Sigma$. 
 \end{itemize}
 \end{remark}

\section{The Language Class $\mathscr{JFA}$}

By the definition of a jumping finite automaton $M$, it is clear that $w\in L_{\jfa}(M)$ implies that $\perm(w)\subseteq L_{\jfa}(M)$, i.\,e., $\perm(L_{\jfa}(M))\subseteq L_{\jfa}(M)$. Since $\perm$ is extensive as a hull operator (see Lemma~\ref{hullOperatorLemma}), we can conclude: 

\begin{corollary}\label{cor-closedsets}
If $L\in \mathscr{JFA}$, then $L$ is perm-closed.
\end{corollary}

This also follows from results of \cite{MedZem2012a}. In particular, we mention the following important  characterization theorem from~\cite{MedZem2014}, that we enrich by combining it with the well-known theorem of Parikh~\cite{Par66} using Proposition~\ref{prop-perm-Parikh}.
\begin{theorem}\label{permCharacterisationTheorem}
$\mathscr{JFA} = \perm(\mathscr{REG})= \perm(\mathscr{CFL})=\perm(\mathscr{PSL})$.
\label{thm-JFA=permREG}
\end{theorem}

This theorem also generalizes the main result of \cite{LatRoz84}. It also indicates that certain properties of this language class have been previously derived under different names; for instance, Latteux~\cite{Lat79a} 
writes $\mathscr{JFA}$ as c(RAT), and he mentions yet another characterization for
this class in the literature, which is the class of all perm-closed languages whose Parikh image is semilinear; 
as explained in Section~\ref{sec:basicDefinitions}, the class of languages whose Parikh image is semilinear is also known as \emph{slip languages} \cite{GinSpa71}, or $\mathscr{PSL}$ for short.
Due to Lemma~\ref{permcharac}, there is a natural bijection between $\mathscr{JFA}$ and the recognizable subsets of the monoid $(\mathbb{N}^\Sigma, +, \vec{0})$. 

Let us mention one corollary that can be deduced from these connections; for proofs, we refer to \cite{EilSch69,GinSpa66}.

\begin{corollary}
$\mathscr{JFA}$ is closed under intersection and under complementation.
\end{corollary}

Notice that the proof given in
Theorem 17.4.6 in \cite{MedZem2014} is wrong, as the nondeterminism inherent in JFAs due to the jumping feature is neglected.
For instance, consider the deterministic\footnote{According to \cite{MedZem2014}, a JFA is \emph{deterministic} if each state has exactly one outgoing transition for each letter.} 
JFA $M=(\{r,s,t\},\{a,b\},R,\{s\},F)$ with rules according to Figure~\ref{fig:ComplCounterexample}.
\begin{figure} 
\begin{centering}
\includegraphics{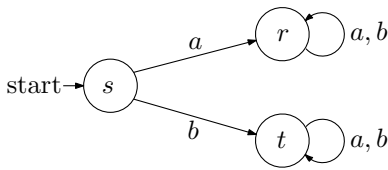}
\par\end{centering}
\caption{\label{fig:ComplCounterexample}An example JFA, final states not specified.}
\end{figure}
If $F=\{r\}$, then $M$ accepts all words  that contain at least one $a$. But, if $F=\{s,t\}$, then $M$ accepts $\varepsilon$ and all words  that contain at least one $b$.
This clearly shows that the standard state complementation technique does not work for JFAs.


\LV{As we will be concerned later also with descriptional and computational complexity issues, let us mention here that, a}ccording to the analysis 
indicated in~\cite{EspGKL2011}, Parikh's original proof would produce,
starting from a context-free grammar $G$ with $n$ variables, a regular expression $E$ of length $\landau\!\left(2^{2^{n^2}}\right)$ such that $\perm(L(G))=\perm(L(E))$, whose corresponding NFA is even bigger, while the construction of~\cite{EspGKL2011} results in an NFA $M$ with only $4^n$ states, satisfying $\perm(L(G))=\perm(L_{\fa}(M))$.\LV{
In the context of this theorem, it is also interesting to note that
recently there have been investigations on the descriptional complexity
of converting context-free grammars into Parikh-equivalent finite automata, see \cite{LavPigSek2013}.
The theorem also links JFAs to the literature on ``commutative context-free languages'', e.\,g.,   \cite{BeaBlaLat87,Kor98}.
}

Also, a sort of normal forms for language classes $\mathcal{L}$ such that $\perm(\mathcal{L})=\mathscr{JFA}$
have been studied, for instance, the class  $\mathcal{L}$ of letter-bounded languages can be characterized in various ways, see \cite{CadFinMcK2012,Choetal2012,IbaSek2012} for a kind of survey.
\par
Since finite languages are regular, we can conclude the following corollary of Theorem~\ref{permCharacterisationTheorem}.

\begin{corollary}\label{cor-fin-closedsets}
Let $L$ be 
a finite language. Then,  $L\in\mathscr{JFA}$ \LV{if and only if}\SV{iff} $L$ is perm-closed.
\end{corollary}

This also shows that all finite $\mathscr{JFA}$ languages are so-called commutative regular languages as studied by Ehrenfeucht, Haussler and Rozenberg in \cite{EhrHauRoz83}. We will come back to this issue later.

Next, we shall show that $\mathscr{JFA}$ coincides with the class of $\alpha$-\textsc{SHUF} expressions. To this end, we first observe that a regular expression $E$ can be easily turned into an $\alpha$-\textsc{SHUF} expression describing $\perm(L(E))$ by replacing catenations and Kleene stars with shuffles and iterated shuffles (this is a direct consequence of the fact that the $\perm$ operator is a semiring morphism as stated in Theorem~\ref{thm-semiring-morph}).

\begin{lemma}\label{lem-reg-alphashuf}
Let $R'$ be a regular expression. Let the $\alpha$-\textsc{SHUF} expression $R$ be obtained from $R'$ by consequently replacing all $\cdot$ by $\shuffle$, and all ${}^*$ by ${}\shufflestar$ in $R'$. Then, $\perm(L(R')) = L(R)$.
\label{lem:31}
\end{lemma}



\begin{proof}Let $R'$ be a regular expression. By definition, this means that $L(R') = K$, where $K$ is some expression over the languages $\emptyset$, $\{\emptyword\}$ and $\{a\}$, $a \in \Sigma$, using only union, catenation and Kleene-star. By Theorem~\ref{thm-semiring-morph}, $\perm(K)$ can be transformed into an equivalent expression $K'$ using only union, shuffle and iterated shuffle. Furthermore, in $K'$, the operation $\perm$ only applies to languages of the form $\emptyset$, $\{\emptyword\}$ and $\{a\}$, $a \in \Sigma$, which means that by simply removing all $\perm$ operators, we obtain an equivalent expression $K''$ of languages $\emptyset$, $\{\emptyword\}$ and $\{a\}$, $a \in \Sigma$, using only union, shuffle and iterated shuffle. This expression directly translates into the $\alpha$-\textsc{SHUF} expression $R$ with $L(R) = \perm(L(R'))$.
\end{proof}



\noindent
We are now ready to prove our characterization theorem for $\mathscr{JFA}$.

\begin{theorem}\label{thm-JFA=a-SHUF}
A language $L\subseteq\Sigma^*$ is in $\mathscr{JFA}$ if and only if there is some  $\alpha$-\textsc{SHUF} expression $R$ such that $L=L(R)$.
\end{theorem}

\begin{proof}
If  $L\in\mathscr{JFA}$, then there exists a regular
language $L'$ such that $L=\perm(L')$ by Theorem~\ref{thm-JFA=permREG}. $L'$ can be described by some regular expression $R'$. By Lemma~\ref{lem:31},
we find an  $\alpha$-\textsc{SHUF} expression $R$ with $L=\perm(L(R'))=L(R)$.

Conversely, if $L$ is described by some $\alpha$-\textsc{SHUF} expression $R$, i.\,e., $L=L(R)$, then
construct the regular expression $R'$ by  consequently replacing all $\shuffle$ by $\cdot$ and
 all $\shufflestar$ by ${}^*$ in $R$. Clearly, we face the situation described in Lemma~\ref{lem:31}, so
 that we conclude that $\perm(L(R')) = L(R)=L$.
 As $L(R')$ is a regular language,  $\perm(L(R'))=L\in\mathscr{JFA}$ by Theorem~\ref{thm-JFA=permREG}.
\end{proof}

Since $\alpha$-\textsc{SHUF} languages are closed under iterated shuffle, we obtain the following corollary as a consequence of Theorem~\ref{thm-JFA=a-SHUF}, adding to the list of closure properties given in \cite{MedZem2012a}.
\begin{corollary}
$\mathscr{JFA}$ is closed under iterated shuffle. 
\label{cor:1}
\end{corollary}

We like to point out again the connections to regular expressions over 
commutative monoids as studied in \cite{EilSch69}.

Let us finally mention a second characterization of the finite perm-closed sets in terms of $\alpha$-\textsc{SHUF} expressions (recall that Corollary~\ref{cor-fin-closedsets} states the first such characterization).

\begin{proposition}\label{prop-fin-closedsets}
Let $L$ be a language. Then,  $L$ is finite and perm-closed if and only if there is an $\alpha$-\textsc{SHUF} expression $R$, with $L=L(R)$, that does not contain the iterated shuffle operator.
\end{proposition}

\begin{proof}Let $L$ be 
a finite language with $L=\perm(L)$. Clearly, there is a regular expression $R_L$, with $L(R_L)=L$, that uses only 
the catenation and union operations.
As  $L$ is perm-closed, the $\alpha$-\textsc{SHUF} expression $R$ obtained from $R_L$ by replacing all
catenation by shuffle operators satisfies
$L(R)=\perm(L(R_L))=L$ by Lemma~\ref{lem:31} and
does not contain the iterated shuffle operator.
Conversely, let  $R$ be an $\alpha$-\textsc{SHUF} expression that does not contain the iterated shuffle operator. By combining
Theorem~\ref{thm-JFA=a-SHUF} with Corollary~\ref{cor-closedsets}, we know that $L(R)$ is perm-closed. 
It is rather straightforward that $L(R)$ is also finite.
\end{proof}

We conclude this section with an example.

\begin{example}\label{ex:2}
Let $M$ be the finite machine presented in Figure~\ref{fig:JFAAlphaShuffExample}. In the standard way, we can turn $M$ into the regular expression
\begin{align*}
E = \:&(( a b^* a b )^* ( (a b^* a a) + b ) ( a b^* a a )^* ( (a b^* a b) + b ) )^*\\
&( a b^* a b )^* ( (a b^* a a) + b ) ( a b^* a a )^* 
\end{align*}
with $L_{\fa}(M) = L(E)$. By Lemma~\ref{lem-reg-alphashuf}, $L_{\jfa}(M) = L(E')$, where
\begin{align*}
E' = \:&(( a \shuffle b\shufflestar \shuffle a \shuffle b )\shufflestar \shuffle ( (a \shuffle b\shufflestar \shuffle a \shuffle a) + b )\\
&\shuffle ( a \shuffle b\shufflestar \shuffle a \shuffle a )\shufflestar \shuffle ( (a \shuffle b\shufflestar \shuffle a \shuffle b) + b ) )\shufflestar\\
&\shuffle ( a \shuffle b\shufflestar \shuffle a \shuffle b )\shufflestar \shuffle ( (a \shuffle b\shufflestar \shuffle a \shuffle a) + b )\\ 
&\shuffle ( a \shuffle b\shufflestar \shuffle a \shuffle a )\shufflestar\,.
\end{align*}
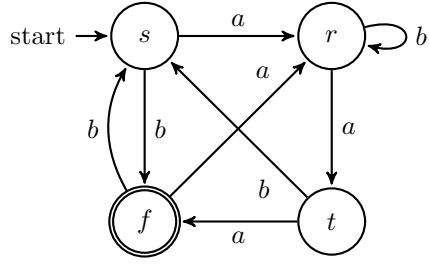
\begin{figure}
\begin{center}
\begin{tikzpicture}[->,>=stealth',shorten >=1pt,auto,node distance=7em, thick]
  \tikzstyle{every state}=[fill=none,draw=black,text=black]
  \node[initial,state] (s)   [fill= none, draw=black,text=black]    {$s$};
  \node[state,accepting] (f) [below of=s]       					{$f$};
  \node[state]         (r)   [right of=s]                           {$r$}; 
  \node[state]         (t)   [right of=f]                           {$t$};
  \path (s) edge node {$a$} (r)
   		(s) edge node {$b$} (f)
        (r) edge node {$a$} (t)
        (r) edge [loop right] node {$b$} (r)
		(t) edge node {$a$} (f) 
        (t) edge [pos=0.2] node {$b$} (s)
        (f) edge [pos=0.8] node {$a$} (r)
        (f) edge [bend left] node {$b$} (s);						
\end{tikzpicture}
\end{center}
\caption{The finite machine of Example~\ref{ex:2}.}
\label{fig:JFAAlphaShuffExample}
\end{figure}
\end {example}

\section{The Language Classes $\mathscr{GJFA}$ and $\mathscr{SHUF}$}


In the last section, we saw that JFA and $\alpha$-\textsc{SHUF} expressions correspond to each other in a very similar way as classical regular expressions correspond to finite automata. More precisely, in the translation between $\alpha$-\textsc{SHUF} expressions and JFA, the atoms of the $\alpha$-\textsc{SHUF} expression will become the labels of the JFA and vice versa. \par
GJFA differ from JFA only in that the labels can be arbitrary words instead of symbols and, similarly, \textsc{SHUF} expressions differ from $\alpha$-\textsc{SHUF} expressions only in that the atoms can be arbitrary words. This suggests that a similar translation between GJFAs and \textsc{SHUF} expressions exists and, thus, these devices describe the same class of languages. Unfortunately, this is not the case, which can be demonstrated with a simple example: let $M = (\{s\}, \{a, b, c, d\}, \{s ab \to s, s cd \to s\}, s, \{s\})$ be a GJFA, which naturally translates into the \textsc{SHUF} expression $E = (ab \shuffle cd)\shufflestar$. It can be easily verified that every word that is accepted by $M$ can also be generated by $E$, but, as $acbd \in (L(E)\setminus L_{\jfa}(M))$, we have $L_{\jfa}(M) \subsetneq L(E)$. \par
In the following, we shall see that not only this naive translation between GJFA and \textsc{SHUF} expressions fails, but the language classes $\mathscr{GJFA}$ and $\mathscr{SHUF}$ are incomparable.

\begin{lemma}\label{GJFAnotSHUF}
Let $M = (\{s\}, \{a, b, c, d\}, \{s ab \to s, s cd \to s\}, s, \{s\})$. Then $L(M)$ is not a \textsc{SHUF} language.
\end{lemma}

\begin{proof}
For the sake of contradiction, let $E$ be a \textsc{SHUF} expression with $L(E)=L(M)$. As the number of occurrences of both $a$ and $d$ grows infinitely in words from $L(M)$, one of the two cases must hold: 
\begin{itemize}
\item $E$ contains a subexpression $(R)\shufflestar$ such that there exists a $w \in L(R)$ with $|w|_{a} \geq 1$ and $|w|_{d} \geq 1$.
\item $E$ contains a subexpression $R_1 \shuffle R_2$ such that there exists a $w \in L(R_1)$ with $|w|_{a} \geq 1$ and a $w' \in L(R_2)$ with $|w|_{d} \geq 1$. 
\end{itemize}
Both cases imply that $L(E)$ contains a word with factor $ad$. This is a contradiction, since such words are not in $L_{\jfa}(M)$.
\end{proof}


\begin{lemma}\label{SHUFnotGJFA}
Let $L=L((ac \shuffle bd)\shufflestar)$. $L$ is not accepted by any GJFA.
\end{lemma}

\begin{proof}
For the sake of contradiction, assume that $L$ is accepted by a GJFA $M$. Let
$n$ be greater than 
the maximum length of a transition label in $M$ and let $w=ab^ncd^n$. The accepting computation of $M$ on $w$ uses exactly one transition with a label $u$ that contains $c$.
\begin{itemize}
\item If $u=b^icd^j$ for $i,j\geq 0$, all earlier transitions only consume factors that are completely contained in the prefix $a b^{n-i}$ or the suffix $d^{n-j}$ of $w = a b^{n-i} (b^icd^j) d^{n-j}$.
This implies that, by using the same sequence of transitions, $M$ can accept $w'=b^icd^jab^{n-i}d^{n-j}$.
\item Otherwise, $u=ab^rcd^s$ for $r,s\geq 0$, i.\,e., it contains both $a$ and~$c$.
By the choice of $n$, an earlier transition labeled with $b^k$ with $k>0$ was used. 
However, this implies that also $w''=ab^{n-k}cd^nb^k$ is
accepted by $M$.
\end{itemize}
The case of $w'\in L(M)$ violates the condition that the symbol $a$ precedes $c$ in words from $L$, while the case of $w''\in L_{\jfa}(M)$ contradicts the fact that the words in $L$ do not end with $b$.
\end{proof}

\begin{lemma}\label{bothSHUFGJFA}
$\{ab\}\shufflestar \in (\mathscr{GJFA} \cap \mathscr{SHUF})\setminus \mathscr{JFA}$.
\end{lemma}

\begin{proof}
Obviously, $\{ab\}\shufflestar = L((ab)\shufflestar)$. Furthermore, $\{ab\}\shufflestar = L_{\jfa}(M)$, where $M$ is the GJFA with a single state $s$, which is both initial and final, and a single rule $sab\to s$.
As $ab\in\{ab\}\shufflestar $, but
$ba \notin \{ab\}\shufflestar$, $\{ab\}\shufflestar$ is not perm-closed, and hence not a JFA language.
\end{proof}

It is interesting to note that if we take the permutation closures of the separating languages from Lemmas~\ref{GJFAnotSHUF}~and~\ref{SHUFnotGJFA}, then we get JFA languages. As shall be demonstrated next (see Theorem~\ref{thm-GJFA-SHUF}), this property holds for all \textsc{SHUF} and GJFA languages.

\begin{lemma}
\label{prop-permclosure-GJFA-SHUF}
$\perm(\mathscr{GJFA})\cup\perm(\mathscr{SHUF})
=\mathscr{JFA}$.
\end{lemma}
\begin{proof}
Clearly $\mathscr{JFA}\subseteq \perm(\mathscr{GJFA})\cup\perm(\mathscr{SHUF})$ and thus we only have to show that $\perm(\mathscr{GJFA})\cup\perm(\mathscr{SHUF})
\subseteq \mathscr{JFA}$.
\begin{itemize}
\item Let $L\in \mathscr{SHUF}$ be described by a SHUF expression $X$. Then $\perm(L)$ is described by the $\alpha$-SHUF expression $X'$
that is obtained from $X$ by replacing each atomic word $a_1\cdots a_n\in\Sigma^*$
of length $n\geq 2$  by the $\alpha$-SHUF subexpression $a_1\shuffle\dots\shuffle a_n$.
The fact that $\perm(L)=\perm(L(X))=L(X')$ follows by an easy induction argument using Theorem~\ref{thm-semiring-morph}.
\item Let $L\in \mathscr{GJFA}$. The well-known construction of a finite automaton that simulates a given general finite automaton can be applied to obtain, from a given GJFA $M$, a JFA $M'$ with the property $\perm(L_{\jfa}(M))=L_{\jfa}(M')$. The correctness of this method immediately follows from our reasoning towards Theorem~\ref{permCharacterisationTheorem}.
\end{itemize}
In both the cases we conclude that $\perm(L)$ lies in $\mathscr{JFA}$.
\end{proof}

 \begin{lemma}
 Let $L\subseteq \Sigma^*$. Then the following claims are equivalent:
\begin{enumerate}
\item $L\in\mathscr{JFA}$,
\item $L$ is perm-closed and $L\in\mathscr{GJFA}$,
\item $L$ is perm-closed and $L\in\mathscr{SHUF}$.
\end{enumerate}
 \end{lemma}
 \begin{proof}
 As each  $L\in\mathscr{JFA}$ is perm-closed and in $\mathscr{GJFA}\cap \mathscr{SHUF}$, we only have to show the upward implications. If  $L\in\mathscr{GJFA}$, then
 (by Lemma~\ref{prop-permclosure-GJFA-SHUF}), $\perm(L)\in\mathscr{JFA}$. If, in addition, $L$ is perm-closed, then
 $\perm(L)=L$, which shows the claim.
 Similarly, we can show that, if $L$ is perm-closed and $L\in\mathscr{SHUF}$, then  $L\in\mathscr{JFA}$.
 \end{proof}
  
\begin{theorem}\label{thm-GJFA-SHUF}
 $\perm(\mathscr{GJFA})=\perm(\mathscr{SHUF})=\perm(\mathscr{PSL})
=\mathscr{JFA}$.
 \end{theorem}
 
 \begin{figure}
\begin{center}
\begin{tikzpicture}[scale=.7]
  \node (psl) at (0,4) {$\mathscr{PSL}$};
  \node (npermpsl) at (-3,-2) {$\perm(\mathscr{GJFA}) = \perm(\mathscr{SHUF}) = $};
  \node (permpsl) at (-3,-3) {$\perm(\mathscr{PSL}) = \mathscr{JFA} = \alpha\mbox{-}\mathscr{SHUF}$};
   \node (shufcapgjfa) at (-3,0) {$\mathscr{SHUF} \cap \mathscr{GJFA}$};
  \node (reg) at (3,-3) {$\mathscr{REG}$};
  \node (shuf) at (-3,2) {$\mathscr{SHUF}$};
  \node (gjfa) at (0,2) {$\mathscr{GJFA}$};
  \node (cfl) at (3,2) {$\mathscr{CFL}$};
  \node (regcapjfa) at (0,-5) {$\mathscr{REG} \cap \mathscr{JFA}$};
  \draw[->] (regcapjfa) -- (permpsl);
  \draw[->] (regcapjfa) -- (reg);
  \draw[->] (reg) -- (cfl);
  \draw[->] (shufcapgjfa) -- (shuf);
  \draw[->] (shufcapgjfa) -- (gjfa);
  \draw[->] (npermpsl) -- (shufcapgjfa);
  \draw[->] (shuf) -- (psl);
  \draw[->] (gjfa) -- (psl);
  \draw[->] (cfl) -- (psl);
 \end{tikzpicture}
\end{center}
\caption{\label{fig:hierarchies}Inclusion diagram of our language families.}
\end{figure}
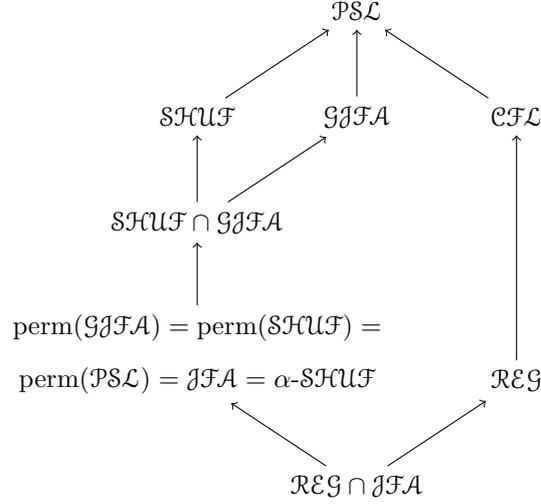
 
 
We summarize the inclusion relations between the language families considered in this paper in Figure \ref{fig:hierarchies}.
In this figure, an arrow from class $A$ to $B$ represents the strict inclusion $A \subsetneq B$.
A missing connection between a pair of language families means incomparability.



\begin{theorem}
The inclusion and incomparability relations 
displayed in Figure~\ref{fig:hierarchies} are correct.
\end{theorem}
 
\begin{proof}  
We first show the correctness of the subset relations. The class $\mathscr{REG} \cap \mathscr{JFA}$ is obviously included in both $\mathscr{REG}$ and $\mathscr{JFA}$, and any non-commutative regular language and the non-regular JFA language $\{ab\}\shufflestar$ show these subset relations to be proper. That $\mathscr{JFA} \subsetneq \mathscr{SHUF} \cap \mathscr{GJFA}$ follows by definition and Lemma~\ref{bothSHUFGJFA}. Similarly, both $\mathscr{SHUF} \cap \mathscr{GJFA} \subsetneq \mathscr{GJFA}$ and $\mathscr{SHUF} \cap \mathscr{GJFA} \subsetneq \mathscr{SHUF}$ follows by definition and Lemmas~\ref{GJFAnotSHUF}~and~\ref{SHUFnotGJFA}, respectively. Theorem~\ref{thm-GJFA-SHUF} shows that the classes $\mathscr{SHUF}$ and $\mathscr{GJFA}$ are contained in $\mathscr{PSL}$ and since $\mathscr{SHUF}$ and $\mathscr{GJFA}$ are incomparable, these inclusions are proper. By Parikh's theorem~\cite{Par66} and as the context-free languages do not contain the language studied in Example~\ref{ex:1}, $\mathscr{CFL} \subsetneq \mathscr{PSL}$. Finally, $\mathscr{REG} \subsetneq \mathscr{CFL}$ is well-known; thus, all the claimed proper subset relations hold.\par
Since $\mathscr{JFA}$ is a proper superclass of $\mathscr{REG} \cap \mathscr{JFA}$, it contains a language not in $\mathscr{REG}$. Furthermore, according to \cite[Lemma 17.3.2]{MedZem2014}, the regular language $\{a\}^*\{b\}^*$ is not in $\mathscr{GJFA}$. By a similar argument as used in the proof of Lemma~\ref{GJFAnotSHUF}, it can also be shown that $\{a\}^*\{b\}^* \notin \mathscr{SHUF}$ (more precisely, since this language is infinite, either a subexpression that contains both $a$ and $b$ is subject to an iterated shuffle operation or two subexpressions that produce only $a$ and $b$, respectively, are connected by a shuffle operation). Hence, $\mathscr{REG}$ is incomparable with all the classes on the left side of the diagram. The language of Example~\ref{ex:1} is in $\mathscr{JFA}$, but not in $\mathscr{CFL}$. Furthermore, $\{a\}^*\{b\}^*$ is a context-free language, which implies that $\mathscr{CFL}$ is also incomparable with all the classes on the left side of the diagram. Finally, the incomparability of the classes $\mathscr{SHUF}$ and $\mathscr{GJFA}$ is established by Lemmas~\ref{GJFAnotSHUF},~\ref{SHUFnotGJFA}~and~\ref{bothSHUFGJFA}. This concludes the proof.
 \end{proof}
 
The inclusion diagram also motivates to study problems that can be expressed as follows. If $\mathscr{X}$ and  $\mathscr{Y}$ are two language families with $\mathscr{X}\subsetneq \mathscr{Y}$ and if  $\mathscr{Y}$ can be described by  $\mathscr{Y}$-devices, what is the complexity (or even decidability) status of the problem, given some  $\mathscr{Y}$-device $Y$, to determine if the language $L(Y)$ belongs to $\mathscr{X}$? In Section~\ref{sec:compareJFAREG}, we shall see that this type of problem is NP-hard for $\mathscr{X}=\mathscr{REG}\cap \mathscr{JFA}$ and $\mathscr{Y}=\mathscr{REG}$ (Theorem~\ref{thm-non-JFA-complexity}) or $\mathscr{Y}=\mathscr{JFA}$ (Theorem~\ref{thm-non-REG-complexity}).
Conversely, it is even undecidable whether or not a given context-free grammar generates a regular language; see~\cite{Gre68}.

\section{Representations and Normal Forms}


One of the main results of the conference version
of this paper~\cite{FerParSch2015} was the following representation theorem.


\begin{theorem}[Representation Theorem]
\label{thm-representation}
Let $L\subseteq\Sigma^*$. Then, $L \in \mathscr{JFA}$ if and only if there exists a number $n \geq 1$ and finite sets $M_i \subseteq \Sigma^*$, $N_i \subseteq \Sigma^*$ for $1 \leq i \leq n$, so that the following representation is valid.
\begin{equation} \label{eq-representation}
L = \bigcup^{n}_{i=1} \perm(M_i) \shuffle (\perm(N_i))\shufflestar
\end{equation}
\end{theorem}

We have sketched a  proof of this representation theorem in~\cite{FerParSch2015}
on the level of $\alpha$-\textsc{SHUF} expressions, 
so that
we actually got a normal form theorem for these expressions.
Our proof idea was similar to the one that Jantzen presented in~\cite{Jan79a}. 
However, in the meantime we understood the connections
to Parikh's theorem better, so that we will present a different reasoning in the following proof.

\begin{proof}
Consider  $L\subseteq\Sigma^*$ with $L \in \mathscr{JFA}$.
By Theorem~\ref{permCharacterisationTheorem}, $\psi_\Sigma(L)$ is semilinear. Hence,
$$\psi_\Sigma(L)=\bigcup_{i=1}^nS_i,
$$
where the sets $S_i$ are linear sets, which means that there are 
vectors $v^i,v^i_1,\dots, v^i_{\ell_i}$ such that
$$S_i=\{x\in \mathbb{N}^{|\Sigma|}\suchthat \exists 
k_1,\dots, k_{\ell_i}: x=v^i+\sum_{j=1}^{\ell_i}k_iv^i_j
\}.$$
Let $M_i=\psi_\Sigma^ {-1}(v^i)$ and $N_{i,j}=\psi_\Sigma^ {-1}(v^i_j)$. Then,
$$S_i=\psi_\Sigma\left(M_i\shuffle \bigshuffle_{j=1}^{\ell_i}N_{i,j}\shufflestar\right)
=\psi_\Sigma\left(M_i\shuffle\left(\bigcup_{j=1}^{\ell_i}N_{i,j}
\right)\shufflestar
\right)
, $$
as $\psi_\Sigma$ acts as a morphism. Let $N_i=\bigcup_{j=1}^{\ell_i}N_{i,j}$. Observe that by our definition, $M_i$, $N_{i,j}$ and hence $N_i$ are all perm-closed. Hence, $L$ can be represented as required.

As the required representation can be easily interpreted as some $\alpha$-\textsc{SHUF} expression, any $L$ that can be represented as in the theorem is in  $\mathscr{JFA}$ according to Theorem~\ref{thm-JFA=a-SHUF}.
\end{proof}

Recall that for classical regular expressions, the star height was quite an important notion; see \cite{Coh70,CohBrz70,Has82}.
Actually, the proof that we sketched in \cite{FerParSch2015} for a proof of the 
preceding theorem was based on an inductive argument involving the star height of the expressions.
Without giving further details, including omitting some further definitions, we only mention the following
interesting consequences from the Representation Theorem.

This can be 
obtained by combining Theorem~\ref{thm-representation} with Theorem~\ref{thm-JFA=a-SHUF}, Lemma~\ref{lem-reg-alphashuf} and Theorem~\ref{thm-JFA=permREG}.

\begin{corollary}$L\in \mathscr{JFA}$ if and only if there is a
regular language $R$ of star height at most one such that $L=\perm(R)$.
\end{corollary}

From Proposition~\ref{prop-fin-closedsets}, we can immediately deduce:

\begin{corollary}\label{finitePermClosedLanguagesCharacterisationCorollary}
A language is finite and perm-closed if and only if it can be described by some $\alpha$-\textsc{SHUF} expression of
height zero.
\end{corollary}

Combining Corollary~\ref{finitePermClosedLanguagesCharacterisationCorollary} with
Theorem~\ref{thm-JFA=permREG} 
and the well-known fact that finiteness of regular expressions can be decided, we immediately obtain the following, as Theorem~\ref{thm-representation} guarantees that the height of $\mathscr{JFA}$ languages is zero or one:

\begin{corollary}
It is decidable, given some JFA and some integer $k$, whether or not this JFA describes 
a language of 
height at most $k$.
\end{corollary}

Notice that we have formulated, in this corollary, the shuffle analogue of the famous star height problem, 
%
which 
has been a major open problem for regular languages~\cite{Has88}.
Recall that Eggan's Theorem \cite{Egg63} relates the star height of a regular language to its so-called cycle rank, which formalizes loop-nesting in NFA's. Again, the characterization theorems that we 
derived allow us to conclude that, in short, for any $L\in\mathscr{JFA}$ there exists some finite machine $M$ of cycle rank at most one such that $L_{\jfa}(M)=L$.

There is actually yet another way to derive the Representation Theorem. Namely, Eilenberg and Schützenberger derived in \cite{EilSch69} regular expressions over free commutative monoids. As also mentioned earlier, this can be again re-interpreted as kind of regular expressions dealing with Parikh vectors. In~\cite{EilSch69}, the connection to the definition of semilinear sets is also drawn, although with a different method and background.

\section{Comparing $\mathscr{JFA}$ and $\mathscr{REG}$}\label{sec:compareJFAREG}

By the results of Meduna and Zemek, we know that $\mathscr{JFA}$ and $\mathscr{REG}$ are two incomparable families of languages.
Above, we already derived several characterizations of 
$\mathscr{JFA}\cap \mathscr{FIN}\subset \mathscr{REG}$.
Let us first explicitly state a characterization of $\mathscr{JFA}\cap \mathscr{REG}$ that can be easily deduced from our previous results.

\begin{proposition}
$L\in \mathscr{JFA}\cap \mathscr{REG}$ iff $L\in \mathscr{REG}$ and $L$ is perm-closed.
\end{proposition}

We mention this, as the class 
$\mathscr{JFA}\cap \mathscr{REG}$ can be also characterized as follows 
according to
 Ehrenfeucht, Haussler and Rozenberg~\cite{EhrHauRoz83}.
 Namely, they describe this class of (what they call) commutative regular languages as finite unions of periodic languages.
 We are not giving a definition of this notion here, but rather 
 state an immediate consequence of their characterization in our terminology.

 \LV{According to
 Ehrenfeucht, Haussler and Rozenberg~\cite{EhrHauRoz83},
 a sequence of vectors 
 $$\rho = v_0,v_1,\dots, v_{|\Sigma|} \in \mathbb{N}^{|\Sigma|}$$
 is called a \emph{base} (with respect to $\Sigma$) iff, for all
 $i,j\in\{1,\dots,|\Sigma|\}$, $v_i(j)=0$
 if $i\neq j$. The \emph{$\rho$-set}, written
 $\Theta(\rho)$, of a base $\rho$ is defined by
 $$\Theta(\rho)=\{v\in\mathbb{N}^{|\Sigma|}\suchthat \exists \ell_1,\dots,\ell_{|\Sigma|}: v=v_0+\sum_{i=1}^{|\Sigma|}\ell_i\cdot v_i\}\,.
 $$
 $\rho$-sets are linear sets, and they are in one-to-one correspondence with their bases in the following sense:
 \begin{lemma}[\cite{EhrHauRoz83}]
 Let $\rho,\rho'$ be bases with respect to $\Sigma$. Then, $\rho=\rho'$ if and only if 
 $\Theta(\rho)=\Theta(\rho')$.
 \end{lemma}
 Now, a language $L\subseteq\Sigma^*$ is called \emph{periodic} iff it is perm-closed and there is a base $\rho$  
 with respect to $\Sigma$ such that $\psi_\Sigma(L)=\Theta(\rho)$.
 
 \begin{proposition}
 Let $L\subseteq\Sigma^*$.
 $L$ is periodic if and only if, for some word $w\in \Sigma^*$ and some function $n:\Sigma\to \mathbb{N}$, 
 $$L=\perm(w)\shuffle \left(\bigcup_{a\in\Sigma}a^{n(a)}\right)\shufflestar\,.$$
 \end{proposition}
  
 \begin{proof}Let $\Sigma=\{a_1,\dots,a_{|\Sigma|}\}$. 
If $L$ is periodic, then $L=\psi_\Sigma^{-1}(\Theta(\rho))$ for some base $\rho=v_0,v_1,\dots,v_{|\Sigma|}$. Select $w\in\psi_\Sigma^{-1}(v_0)$ and set $n(a_i)=v_i(i)$.
Then, $L = \left(\perm(w)\shuffle \left(\bigcup_{a\in\Sigma}a^{n(a)}\right)\shufflestar\right)$. 
Conversely, given $L=\perm(w)\shuffle \left(\bigcup_{a\in\Sigma}a^{n(a)}\right)\shufflestar\,,$ one can see that 
$v_0=\psi_\Sigma(w)$, $v_i(i)=n(a_i)$ and
$v_i(j)=0$ for $i\neq j$ defines a base $\rho$ such that $\psi_\Sigma(L)=\Theta(\rho)$.

 \end{proof}
 
 Ehrenfeucht, Haussler and Rozenberg~\cite{EhrHauRoz83} have shown a characterization theorem that easily yields the following result:
 \begin{corollary}A language $L$ is regular and perm-closed if and only if $L$ is the finite union of periodic languages.
 \end{corollary}
 \begin{proof}If $L$ is regular and perm-closed, then $L$ is  the finite union of periodic languages according to \cite[Theorem 6.5]{EhrHauRoz83}.
 Conversely, as the finite union of perm-closed languages is perm-closed, we can conclude from \cite[Theorem 6.5]{EhrHauRoz83} that  the finite union of periodic languages is regular and perm-closed.
 \end{proof}
 
 The last two results immediately yield Theorem~\ref{thm-representation-2}.
 }
 
\begin{theorem}
\label{thm-representation-2}Let $L\subseteq\Sigma^*$. Then, 
$L \in \mathscr{JFA}\cap \mathscr{REG}$ if and only if there exists a number $n \geq 1$, 
words $w_i$ and finite sets $N_i$ for $1 \leq i \leq n$,
where each $N_i$ is given as $\bigcup_{a\in\Sigma_i}a^{n_i(a)}$
for some $\Sigma_i\subseteq \Sigma$ and some $n_i:\Sigma_i\to \mathbb{N}$,
so that the following representation is valid.
\begin{equation*} 
L = \bigcup^{n}_{i=1} \perm(w_i) \shuffle (\perm(N_i))\shufflestar
\end{equation*}
\end{theorem}

Let us finally mention that yet another characterization of $\mathscr{JFA}\cap \mathscr{REG}$  was derived in \cite[Theorem 3]{LatRoz84}. \LV{Moreover, a relaxed version of the notion of commutativity (of languages) allows a characterization of $\mathscr{REG}$, as shown by Reutenauer~\cite{Reu81}.}
 We would also like to point the reader to \cite{GomAlv2008}, where not only learnability questions of this class of languages were discussed, but also two further normal form representations of $\mathscr{JFA}\cap \mathscr{REG}$ were mentioned.

Next, we consider the problem to decide, for a given JFA or NFA, whether it accepts a language from $\mathscr{JFA}\cap \mathscr{REG}$. This is equivalent to the task of deciding whether a given JFA accepts a regular language\footnote{This question was explicitly asked in the Summary of Open Problems section (ii) of \cite{MedZem2014}.}
or whether a given NFA accepts a commutative language. We shall show that both these problems are co-$\npclass$-hard, even if restricted to automata with binary alphabets. 
Note that for a language over a one-letter alphabet, regularity is equivalent to membership in $\mathscr{JFA}$, so the problem becomes trivial.
\par
We will actually use nearly the same construction for both hardness results, which is based on~\cite{StoMey73}, in which Stockmeyer and Meyer showed how to construct for each 3SAT formula $\phi$ a regular expression $E_\phi$ over the unary alphabet $\Sigma=\{a\}$ such that 
$$\phi\text{ is satisfiable} \text{ if and only if } L(E_\phi)\neq \Sigma^*\,.$$
We will also make use of another property 
$$\phi\text{ is satisfiable} \text{ if and only if } \Sigma^*\setminus L(E_\phi) \text{ is infinite}\,.  $$
%

\begin{theorem}\label{thm-non-REG-complexity}
The non-regularity problem for JFA is $\npclass$-hard, even for binary alphabets.
\end{theorem}

\begin{proof}
We present a reduction from 3SAT. Let $\phi$ be a 3SAT formula and let $E_\phi$ be the regular expression over $\{a\}$ with the properties described above. Let
$$L_\phi=(\{b\}\shufflestar \shuffle L(E_\phi)) \cup  (\{a\} \shuffle \{b\})\shufflestar\,.$$
Note that a finite machine $M_\phi$ with $L_{\jfa}(M_\phi) = L_\phi$ can be easily obtained by transforming $(b \shufflestar \shuffle \widehat{E}_\phi) \cup  (a \shuffle b)\shufflestar$ (where $\widehat{E}_\phi$ is obtained from $E_\phi$ by replacing all catenations and Kleene stars by shuffles and iterated shuffles, respectively) into a finite machine (treating shuffle and iterated shuffle as catenation and Kleene star, respectively) by a standard contruction, e.\,g., the Thompson NFA construction.
\par
Clearly, if $\phi$ is unsatisfiable, then 
$$L_\phi=L((\{b\}\shufflestar\shuffle \{a\}\shufflestar) \cup (\{a\}\shuffle \{b\})\shufflestar) = \{a,b\}^*\,;$$
and thus $L_\phi$ is regular. 
However, if $\phi$ is satisfiable, then $L' = \{a\}^*\setminus L(E_\phi)$ is an infinite regular set. 
Assume, for the sake of contradiction, that  
$L_\phi$ is 
regular. 
Then also 
\begin{align*}
&L_\phi \cap ( \{b\}\shufflestar \shuffle L')\\
=\: &L((\{a\} \shuffle \{b\})\shufflestar) \cap (\{b\}\shufflestar \shuffle L')\\
=\: &\{w\in\{a,b\}^*: |w|_a=|w|_b\land a^{|w|_a}\notin L(E_\phi) \}
\end{align*}
would be regular. However, for every $k, k' \in \mathbb{N}$ with $k \neq k'$ and $a^k, a^{k'} \notin L(E_\phi)$, the words $a^k$ and $a^{k'}$ are not Nerode-equivalent with respect to $L_\phi \cap ( \{b\}\shufflestar \shuffle L')$. Thus, there are infinitely many equivalence classes of this Nerode relation, which is a contradiction.
\end{proof}

Notice that the regularity problem 
is decidable, as shown in a far more general context by Sakarovitch~\cite{Sak92}, referring to older papers of Ginsburg and Spanier~\cite{GinSpa64,GinSpa66}. It would be also interesting to better understand the precise complexities for the regularity problems mentioned by Sakarovitch in the context of trace theory.

It would be of course interesting to determine the exact complexity status of this problem.
Currently, we only know about the mentioned decidability result, which in itself is not completely natural, as similar problems like the universality for flow expressions is known to be undecidable~\cite{Iwa83}.
Next, we deal with the problem of deciding whether a given NFA accepts a commutative language, i.\,e., a language from $\mathscr{JFA}\cap \mathscr{REG}$.
\begin{theorem}\label{thm-non-JFA-complexity}
It is $\npclass$-hard to decide, for a  given NFA $M$, whether $L_{\fa}(M)$ is noncommutative, even for binary languages.
\end{theorem}

\begin{proof}
We use the fact that regular languages are closed under shuffle (see \cite{FliKud2012a} or \cite[p.108] {Gin66}). 
Let $\phi$ be a 3-CNF formula and let $E_\phi$ be the regular expression over $\{a\}$ with the properties described above. We define the language
$$L'_\phi= (\{b\}^* \shuffle L(E_\phi)) \cup \{a\}^* \{b\}.$$
Obviously, there is a finite automaton $M'_\phi$, which can be easily constructed, that accepts $L'_\phi$. If $\phi$ is unsatisfiable, then 
\begin{equation*}
L'_\phi= (\{b\}^* \shuffle \{a\}^*) \cup \{a\}^* \{b\} = \{a,b\}^*
\end{equation*}
is commutative. If $\phi$ is satisfiable, then, for some $a^k \notin L(E_\phi)$, we have $a^k b\in L'_\phi$ and $ba^k\notin L'_\phi$; thus $L'_\phi$ is not commutative.
\end{proof}
Again, we are not aware of a matching upper bound. At least, decidability can be shown as follows. In \cite{KliPol2012}, an explicit construction of a \emph{biautomaton}\footnote{We do not introduce biautomata in this paper.} accepting $L_{\fa}(M)$ for a given DFA $M$ is shown, though there is an exponential blowup in the number of states. Moreover, in \cite{athHOL1} the authors present an algorithm for turning a biautomaton into the \emph{canonical biautomaton} and show that commutativity of a language can be deduced from basic properties of the corresponding canonical biautomaton.




Notice that Theorem~\ref{thm-tally-NFA-universality} from the Appendix
allows us to deduce the following two corollaries. Indeed, because the number of states 
of the constructed automata is only a constant off from the number of states 
of the automaton $M_G$
obtained in the proof of  Theorem~\ref{thm-tally-NFA-universality}, 
this unary NFA could replace the regular expression $E_\phi$ 
used above. 
In fact, that proof (not delivered
in the appendix) was from \textsc{3-Coloring}, and we inherit the following property:
$G$ is 3-colorable if and only if $\{a\}^*\setminus L(M_G)$ is infinite.
This property is important in the reasonings of the hardness proofs shown above.

\begin{corollary}\label{Cor:regularityETH}
There is no algorithm that solves
the regularity problem for $q$-state JFAs on  binary input alphabets
in time $\landau^*(2^{o(q^{1/3})})$, unless ETH fails.
\end{corollary}

\begin{corollary}\label{Cor:commutativityETH}
There is no algorithm that solves
the commutativity problem for $q$-state NFAs on  binary input alphabets
in time $\landau^*(2^{o(q^{1/3})})$, unless ETH fails.
\end{corollary}

If we used the construction of Stockmeyer and Meyer (directly), we would only get bounds worse than $\landau^*(2^{o(q^{1/4})})$ (for more details, we refer to the Appendix).

\section{Complexity Issues\label{Sec:complexity}}
\subsection{Parsing}

For a fixed JFA $M$ and a given word $w\in \Sigma^*$, we can decide whether $w \in L(M)$ in the following way\footnote{In the whole Section \ref{Sec:complexity}, we use $L(M)$ instead of $L_{\jfa}(M)$.}. Scan over $w$ and construct the Parikh mapping $\pi_\Sigma(w)$ of $w$. Simulate a computation of $M$ on $w$ by repeated nondeterministic choice of an outgoing transition, passing to the target state, and decrementing the component of $\pi_\Sigma(w)$ that corresponds to the label $x\in \Sigma$ of the chosen transition. If an accepting state is reached and all the components of $\pi_\Sigma(w)$ are $0$, then $w \in L(M)$. In this procedure, we only have to store the current state and the Parikh mapping, which only requires logarithmic space. Thus, this shows $\mathscr{JFA} \subseteq \nlclass \subseteq \pclass$.\footnote{We wish to point out that this also follows from results in~\cite{CrePie2014}, where containment in $\nlclass$ is shown for a superclass of $\mathscr{JFA}$. 
%
}\par
These considerations show that the \emph{fixed} word problem can be solved in polynomial time. In contrast to the fixed word problem, the \emph{universal} word problem is to decide, for a given automaton $M$ and a given word $w$, whether $w \in L(M)$. The universal word problem for JFA is known to be solvable in polynomial time for fixed alphabets (see, e.\,g.,~\cite{JedSze2001}),
but $\npclass$-complete in general. The hardness follows from our Theorem \ref{JFAUniversalWordProblemETHTheorem} or, e.\,g., from \cite[Theorem 5.1]{MaySto94}, which gives a proof of the NP-hardness (concerning expressions using only union and shuffle) by a reduction from the problem of 3-dimensional matching. Alternatively, a very simple reduction from the Hamiltonian circle problem was given in \cite{Kop2015}. The membership of this problem in NP is shown in our Theorem \ref{GJFAsubNP} or, e.\,g., in \cite{Kop2015} again. See also \cite{Huy83} for generalizations towards commutative context-free grammars.
\par
We shall improve the hardness result 
by giving a reduction from 3\textsc{SAT}. This allows us to conclude a lower bound for the complexity of an algorithm solving the universal word problem for JFA, assuming the exponential time hypothesis (ETH), which is reproduced and further commented in the Appendix.

\begin{figure} 
\begin{centering}
\includegraphics{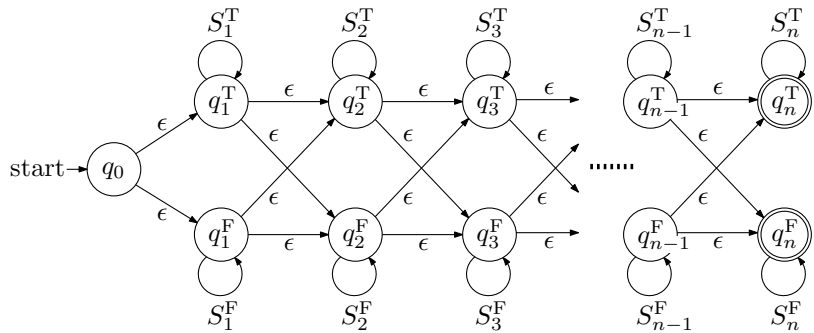}
\par\end{centering}
\caption{\label{fig:univNPcJFA}The JFA $M$ representing a formula $\phi$.}
\end{figure}

\begin{theorem}\label{JFAUniversalWordProblemETHTheorem}
Unless ETH fails, there is no algorithm that, for a given JFA $M$ with state set $Q$ and a given word $w$, decides whether $w \in L(M)$ and runs in time $\landau^*(2^{o(\left|Q\right|)})$.
\end{theorem}

\begin{proof}
Fix a 3-CNF formula $\phi=\bigwedge_{j=1}^{m} C_j$, where each $C_j$ is a disjunction of three literals over variables $x_1,x_2,\ldots,x_n$. Let $\Sigma=\{c_1,\ldots,c_m\}$. For each $i\in\{1\ldots n\}$, let $S_i^\mathrm{T}=\{c_j\suchthat x_i\in C_j\}$ and $S_i^\mathrm{F}=\{c_j\suchthat \neg x_i\in C_j\}$. We claim that the JFA $M=(Q,\Sigma,R,q_0,F)$ with
\begin{eqnarray*}
Q & = & \{q_0\}\cup\{q_1^\mathrm{T},\ldots,q_n^\mathrm{T}\}\cup\{q_1^\mathrm{F},\ldots,q_n^\mathrm{F}\}, \\
F & = & \{q_n^\mathrm{T},q_n^\mathrm{F}\},
\end{eqnarray*}
and transitions according to Figure~\ref{fig:univNPcJFA} accepts the word $w=c_1 c_2 \ldots c_m$ if and only if $\phi$  is satisfiable. 
\begin{itemize}
\item First, let $(\xi_1,\ldots,\xi_n)\in\{\mathrm{T},\mathrm{F}\}^n
$ be an assignment of $x_1,\ldots,x_n$ that satisfies $\phi$. Consider the path in $M$ that uses $\varepsilon$-transitions to visit the states $$q_0, q_1^{\xi_1}, q_2^{\xi_2}, \ldots,q_n^{\xi_n}$$ and, moreover, in each state uses all possible loops (i.\,e., loops labeled by letters that still appear within the input). Because each clause contains some $x_i$ with $\xi_i=\mathrm{T}$ or $\neg x_i$ with $\xi_i=\mathrm{F}$, it follows that each letter $c_j$ of the word $w$ lies in $S_i^{\xi_i}$ for some $i$ an thus is consumed by the loop on $q_i^{\xi_i}$.
\item Second, assume that $w$ is an accepting computation of $M$ on $w$. For each $i\in\{1,\ldots,n\}$ the corresponding path in $M$ must visit exactly one of the vertices $q_i^\mathrm{T},q_i^\mathrm{F}$; let $\xi_i=\mathrm{T}$ or $\xi_i=\mathrm{F}$ respectively. For each $j\in\{1,\ldots,m\}$, the letter $c_j$ is consumed from $w$ and thus lies in $S_i^{\xi_i}$ for some $i$. It follows that each clause contains a satisfied literal.
\end{itemize}
As $\left|Q\right|=2n+1$ and the construction of $M$ works in linear time, any algorithm that solves universal word problem for JFA in time $\landau^*(2^{o(\left|Q\right|)})$ violates ETH.
\end{proof}

For \emph{general} jumping finite automata the complexity of word problems increase considerably. In fact, there is a fixed general jumping finite automaton that accepts an $\npclass$-complete language, i.\,e., for GJFA even the fixed word problem is $\npclass$-complete. \par
Before proving that, let us give the following simple lemma, which is later used for reducing  alphabet sizes of GJFAs.
\begin{lemma}\label{ReducingGjfaAlph}
Let $M$ be a GJFA over $\Sigma=\{x_1,\dots,x_k\}$. Then there exists a homomorphism $h:\Sigma \rightarrow \{0,1\}^*$ and a GJFA $M'$ over $\{0,1\}$ such that, for each $w\in\Sigma^*$,  $w\in L(M)$ if and only if $h(w)\in L(M')$.
\end{lemma}
\begin{proof}
For each $1\le i \le k$, let $h(x_i)=10^i1$. Let $M'$ be obtained from $M$ by replacing each rule $(q,u,r)\in R$ with $(q,h(u),r)$. Clearly, if $w\in L(M)$, then $h(w)\in L(M')$. On the other hand, the definition of $h(x_1),\ldots,h(x_k)$ implies that a computation of $M$ on $h(w)$ can only consume factors of the form $h(x)$ corresponding to particular occurrences of $x$ in $w$.
\end{proof}

\begin{theorem}\label{GJFAsubNP}
$\mathscr{GJFA} \subseteq \npclass$.
\end{theorem}
\begin{proof}
For each GJFA $M=(Q,\Sigma,R,s,F)$ and $w\in L(M)$, there exists a computation of $M$ that consists of at most $\left|Q\right|\left|w\right|$ steps (at most $\left|Q\right|$ transitions labeled by $\varepsilon$ are taken between any two steps that shorten the current word). Thus, the trace of configurations (pairs from $\Sigma^* \times Q$) that leads to acceptance of $w$ has total length at most $\left|Q\right|\left|w\right|^2$. Such a witness for accepting $w$ can be easily checked in polynomial time.
\end{proof}
\begin{theorem}\label{GJFAFixedWordProbHardTheorem}
There exists a GJFA $M$ over a binary alphabet such that $L(M)$ is $\npclass$-complete.
\end{theorem}
\begin{proof}
Let $M=(Q,\{0,1,\overline{0},\overline{1},\star \},R,q_\mathrm{C},\{q_\mathrm{C}\})$ with $Q=\{q_\mathrm{C},q_\mathrm{D},q_0,q_1\}$ be defined according to Figure \ref{fig:NPcGJFA}.
\begin{figure}
\begin{centering}
\includegraphics{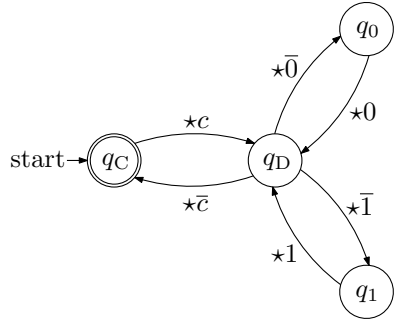}
\par\end{centering}
\caption{\label{fig:NPcGJFA}A GJFA $M$ that can solve the problem $\exactbcov2$.}
\end{figure}
For $u \in \{0, 1\}^*$, $\overline{u}$ is obtained from $u$ by replacing $0$ with $\overline{0}$ and $1$ with $\overline{1}$.\par
Let $u_1, u_2, \ldots, u_n, v \in \{0, 1\}^*$ and $t_i = \star^{|u_i|+2} c \overline{u}_i \overline{c}$, $1 \leq i \leq n$. First, we prove the equivalence of the following two statements. 
\begin{enumerate}
\item\label{enumOne} On input $w = \star^{|v|} v t_1 t_2 \ldots t_n$, $M$ can reach state $q_\mathrm{C}$ from state $q_\mathrm{C}$ with remaining input $w'$ and without visiting $q_\mathrm{C}$ in between.
\item\label{enumTwo} $w' = \star^{|v'|} v' t_1 t_2 \ldots t_{i - 1} t_{i + 1} \ldots t_n$ with $v = u_i v'$, for some $i$, $1 \leq i \leq n$.
\end{enumerate}
Assume that $M$ starts in $q_\mathrm{C}$ with input $w = \star^{|v|} v t_1 t_2 \ldots t_n$. In the first step, while changing into state $q_{\mathrm{D}}$, $M$ consumes a factor $\star c$ from $w$. After this step, the remaining input is $\star^{|v|} v t_1 t_2 \ldots t_{i - 1} \star^{|u_i|+1} \overline{u}_i \overline{c} t_{i + 1} \ldots t_n$, for some $i$, $1 \leq i \leq n$. Now, by using states $q_0$ and $q_1$, a sequence of factors $\star \overline{y}_1$, $\star y_1$, $\star \overline{y}_2$, $\star y_2$, $\ldots$, $\star \overline{y}_m$, $\star y_m$ is consumed, where $y_i \in \{0, 1\}$ for $1 \leq i \leq m$. All these factors only occur in the middle of the factor $\star^{|u_i|} \overline{u}_i$. Furthermore, $M$ can only change into state $q_{\mathrm{C}}$ again if there exists a factor $\star \overline{c}$, which is only the case when the whole factor $\star^{|u_i|} \overline{u}_i$ is consumed. This implies the second statement.\par
If the second statement holds, then the transitions described above (each $y_i$ being chosen such that $y_1 y_2 \ldots y_n = u_i$) will lead $M$ on input $\star^{|v|} v t_1 t_2 \ldots t_n$ from state $q_\mathrm{C}$ into state $q_\mathrm{C}$ without visiting $q_\mathrm{C}$ in between. \par
Next, consider the the following computational problem, which was shown to be $\npclass$-complete in~\cite{JiaSXXZZ2014}: 
\medskip\\
\parbox{1\columnwidth}{
\noindent \textsc{Binary Exact Block Cover} $(\exactbcov2)$ \\ \noindent
\emph{Instance}: Words $u_1, u_2, \ldots, u_k$, and $v$ over $\{0,1\}$.\\
\emph{Question}: Does there exist a permutation $\pi : \{1, 2, \ldots, k\} \to \{1, 2, \ldots, k\}$ such that $v = u_{\pi(1)} u_{\pi(2)} \ldots u_{\pi(k)}$?
}\medskip \\
Let $(u_1, u_2, \ldots, u_k, v)$ be an instance of $\exactbcov2$. If we apply the claim from above inductively, it follows immediately that $\star^{|v|} v t_1 t_2 \ldots t_n \in L(M)$ if and only if there exists a permutation $\pi$ with $v = u_{\pi(1)} u_{\pi(2)} \ldots u_{\pi(k)}$. The permutation $\pi$ corresponds to the order in which the factors $t_i$ are consumed by $M$.\par
Finally, Lemma~\ref{ReducingGjfaAlph} says that $M$ can be easily turned into a binary GJFA $M'$, while the corresponding homomorphism $h$ serves as a polynomial-time reduction from $L(M)$ to $L(M')$.
\end{proof}

Obviously, Theorem~\ref{GJFAFixedWordProbHardTheorem} implies that the universal word problem for GJFA is $\npclass$-complete as well, which has been shown by a separate reduction in the conference version of this paper \cite{FerParSch2015}. We wish to point out, however, that the hardness result for the universal word problem given in \cite{FerParSch2015} is stronger in the sense that it also holds under the restriction to GJFA that accept finite languages. 

\begin{theorem}[\cite{FerParSch2015}]\label{NPCompletenessTheorem}
The universal word problem is $\npclass$-complete for GJFAs accepting finite languages over binary alphabets.
\end{theorem}

A benefit of Theorem~\ref{GJFAFixedWordProbHardTheorem} is that the employed GJFA is rather simple; thus, we obtain a simple proof. 
However, by choosing a more complicated GJFA, we can obtain a reduction from 3\textsc{SAT} to the fixed word problem for GJFA, which allows us to conclude a stronger lower bound that relies on the exponential time hypothesis.

\begin{theorem}\label{GJFAFixedWordProbHardThreeSatTheorem}
There exists a GJFA $M$ over a binary alphabet such that, unless ETH fails, there is no algorithm that, for a given word $w$, decides whether $w \in L(M)$ and runs in time $\landau^*\!\left(2^{o\left({\frac{\left|w\right|}{\log\left|w\right|}}\right)}\right)$.
\end{theorem}
\begin{proof}
Consider the GJFA $M=(Q,\Sigma,R,q_\mathrm{A},\{q_\mathrm{E}\})$, where
\begin{eqnarray*}
\Sigma & = & \{0,1,\overline{0},\overline{1},c_T,c_F,\overline{c},\star,\#, \overline{\star}, \overline{\#}\}, \\
Q & = & \{q_\mathrm{A},q_\mathrm{B}^\mathrm{T},q_\mathrm{B}^\mathrm{F},q_\mathrm{C}^\mathrm{T},q_\mathrm{C}^\mathrm{F},q_0^\mathrm{T},q_0^\mathrm{F},q_1^\mathrm{T},q_1^\mathrm{F},q_\mathrm{D},q_\mathrm{E},q_\mathrm{F},q_\mathrm{G}\},
\end{eqnarray*}
as defined in Figure \ref{fig:GJFA3SAT}. 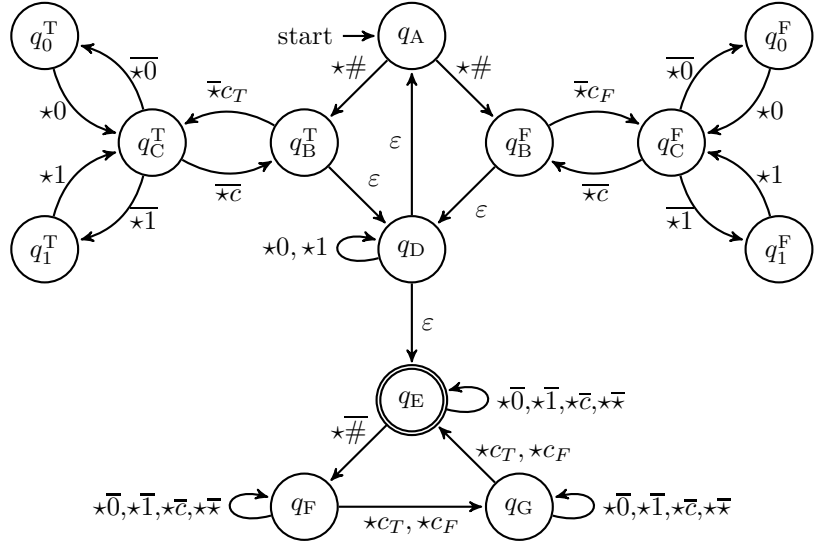
\begin{figure}
\begin{center}
\begin{tikzpicture}[->,>=stealth',shorten >=1pt,auto,node distance=2.0cm, scale = 1, thick]
  \tikzstyle{every state}=[fill=white,draw=black,text=black]
  \node[state]  (q_0)                    {$q_\mathrm{C}^\mathrm{T}$};
  \node[state]  (q_1)     [right of=q_0] {$q_\mathrm{B}^\mathrm{T}$};
  \node[initial,state] (q_2) [above right of=q_1] {$q_\mathrm{A}$};
  \node[state]          (q_3) 	  [below right of=q_2] {$q_\mathrm{B}^\mathrm{F}$};
  \node[state]          (q_4) 	  [right of=q_3] {$q_\mathrm{C}^\mathrm{F}$};
  \node[state]  (q_5)     [below right of=q_1] {$q_\mathrm{D}$};
  \node[state,accepting](q_6) 	  [below of=q_5] {$q_{\mathrm{E}}$};
  \node[state]  (q_7)     [below left of=q_6] {$q_{\mathrm{F}}$};
  \node[state]  (q_8)     [below right of=q_6] {$q_{\mathrm{G}}$};   
   \node[state]  (q_9)    [above left of=q_0] {$q_0^{\mathrm{T}}$};   
    \node[state]  (q_10)     [below left of=q_0] {$q_1^{\mathrm{T}}$};   
       \node[state]  (q_11)     [above right of=q_4] {$q_0^{\mathrm{F}}$};   
    \node[state]  (q_12)     [below right of=q_4] {$q_1^{\mathrm{F}}$}; 
  \path (q_0) edge [bend right] node [below] {$\overline{\star c}$} (q_1)
  		(q_1) edge [bend right] node [above] {$\overline{\star}c_T$} (q_0)
        (q_3) edge [bend left] node [above] {$\overline{\star}c_F$} (q_4)
  		(q_4) edge [bend left] node [below] {$\overline{\star c}$} (q_3)
        (q_0) edge [bend right] node [anchor=center, right] {$\overline{\star 0}$} (q_9)
  		(q_9) edge [bend right] node [anchor=center, left] {$\star 0$} (q_0)
        (q_0) edge [bend left] node [anchor=center, right] {$\overline{\star 1}$} (q_10)
  		(q_10) edge [bend left] node [anchor=center, left] {$\star 1$} (q_0)
        (q_4) edge [bend left] node [anchor=center, left] {$\overline{\star 0}$} (q_11)
  		(q_11) edge [bend left] node [anchor=center, right] {$\star 0$} (q_4)
        (q_4) edge [bend right] node [anchor=center, left] {$\overline{\star 1}$} (q_12)
  		(q_12) edge [bend right] node [anchor=center, right] {$\star 1$} (q_4)
		(q_2) edge node [above] {$\star\#\ \ $} (q_1)
        (q_1) edge node {$\varepsilon$} (q_5)
        (q_5) edge node {$\varepsilon$} (q_2)
        (q_5) edge [loop left] node {$\star 0,\star 1$} (q_5)
        (q_3) edge node {$\varepsilon$} (q_5)
        (q_2) edge node [above] {$\ \ \star\#$} (q_3)
        (q_5) edge node {$\varepsilon$} (q_6)
        (q_6) edge node [anchor=center, above] {$\star\overline{\#}\ \ $} (q_7)
        (q_8) edge node [anchor=center, right] {$\star c_T , \star c_F$} (q_6)
        (q_7) edge node [below] {$\star c_T , \star c_F$} (q_8)
        (q_7) edge [in=165,out=195, loop] node [anchor=center, left] {$\star \overline{0}$,$\star \overline{1}$,$\star \overline{c}$,$\star \overline{\star}$} (q_7)
        (q_8) edge [in=15,out=345,loop] node [anchor=center, right] {$\star \overline{0}$,$\star \overline{1}$,$\star \overline{c}$,$\star \overline{\star}$} (q_8)
        (q_6) edge [in=15,out=345,loop] node [anchor=center, right] {$\star \overline{0}$,$\star \overline{1}$,$\star \overline{c}$,$\star \overline{\star}$} (q_6);
\end{tikzpicture} 
\end{center}
\caption {A GJFA $M$ that can solve 3-SAT.}
\label{fig:GJFA3SAT}
\end{figure}
Fix a 3-CNF formula $\phi=\bigwedge_{j=1}^{m} C_j$, where $$C_j=\lambda_{j,1}\vee\lambda_{j,2}\vee\lambda_{j,3}$$ and  $\lambda_{j,1},\lambda_{j,2},\lambda_{j,3}$ are literals over variables $x_1,x_2,\ldots,x_n$. Suppose that for each $i\in\{1,\ldots,n\}$, the variable $x_i$ occurs $p_i$ times in $\phi$. Note that $\sum_{i=1}^n p_i=3m$.
Let $L=\left\lceil \log_2(n)\right\rceil$ and fix distinct codes $u_1,\ldots,u_n\in\{0,1\}^L$ for the variables. Let 
\begin{eqnarray*}
w_\mathrm{aux} & = & \star^{n+3mL}\# u_1^{p_1}\# u_2^{p_2}\cdots\# u_n^{p_n}, \\
w_\phi & = & \star^{m+m\cdot 6(L+2)} \overline{\#}t_1\overline{\#}t_2\cdots \overline{\#}t_m,
\end{eqnarray*}
where
\begin{eqnarray*}
t_{j,r} & = & \begin{cases}
c_\mathrm{T}\overline{u}_i \overline{c} &\mbox{if }\lambda_{j,r}=x_i,\\
c_\mathrm{F}\overline{u}_i \overline{c} &\mbox{if }\lambda_{j,r}=\neg x_i,
\end{cases} \\
t_j & = & \overline{\star}^{L+2}t_{j,1}\overline{\star}^{L+2}t_{j,2}\overline{\star}^{L+2}t_{j,3}
\end{eqnarray*}
for each $j\in \{1,\ldots,m\}$ and $r\in \{1,2,3\}$. Finally, let
$w=w_\mathrm{aux}w_\phi$ and let us claim that $M$ accepts the word $w$ if and only if $\phi$ is satisfiable. \par
It is clear that the machine works in two phases:
\begin{enumerate}
\item In the \emph{first phase}, which ends once the transition from $q_\mathrm{D}$ to $q_\mathrm{E}$ is taken, some parts of both $w_\mathrm{aux}$ and $w_\phi$ are consumed.
\item In the \emph{second phase}, only the transitions between the states $q_\mathrm{E}$, $q_\mathrm{F}$, and $q_\mathrm{G}$ can be used. Observe:
\begin{itemize}
\item Each of the transition labels contains $c_\mathrm{T}$, $c_\mathrm{F}$, or a letter with bar, which implies that only factors from $w_\phi$ are consumed in the second phase. 
\item Each of the transition labels starts with $\star$, which implies that the remainder of $\overline{\#}t_1\overline{\#}t_2\cdots \overline{\#}t_m$ is consumed left-to-right only.
\item The occurrences of $\overline{\#}$, $c_\mathrm{T}$, and $c_\mathrm{F}$ in transition labels imply that the second phase is successful only if $c_\mathrm{T}$ and $c_\mathrm{F}$ together occur at most twice between successive occurrences of $\overline{\#}$. 
\end{itemize}
It follows that before the second phase starts, at least one of the three occurrences of $c_\mathrm{T}$ and $c_\mathrm{F}$ must be consumed from each of the segments $t_1,t_2,\ldots,t_m$. This corresponds to at least one literal of each clause being satisfied. 
\end{enumerate}
It remains to check that:
\begin{itemize}
\item A run of the first phase must follow some fixed asignment of variables while consuming parts of $t_1,t_2,\ldots,t_m$ (i.\,e., only factors $t_{j,r}$ standing for satisfied literals are consumed).
\item Vice versa, for each assignment there is a run of the first phase that delete all the factors $t_{j,r}$ standing for satisfied literals. 
\end{itemize}
Suppose that $M$ is in the state $q_\mathrm{A}$. It must use a transition labeled with $\star\#$ leading to $q_\mathrm{B}^\mathrm{X}$ for $\mathrm{X}\in \{\mathrm{T},\mathrm{F}\}$, which implies that the remainder of $w_\mathrm{aux}$ is then of the form $\star\cdots\star u_i^{p_i}\# u_{i+1}^{p_{i+1}}\cdots\# u_n^{p_n}$. Then, before passing to $q_\mathrm{D}$ it can repeat the following process up to $p_i$ times:
\begin{itemize}
\item \emph{Open} a literal in any clause by consuming $\overline{\star}c_\mathrm{X}$ from $\overline{\star}^{L+2}t_{j,r}$. If $\mathrm{X}=\mathrm{T}$ (or $\mathrm{X}=\mathrm{F}$), only a positive (negative, respectively) literal can be open.
\item Use the transitions between $q_0^\mathrm{X}$, $q_1^\mathrm{X}$ and $q_\mathrm{C}^\mathrm{X}$ to consume $\overline{\star}^L \overline{u}_i$ from $\overline{\star}^{L+2}t_{j,r}$ together with consuming $\star^L u_i$ from $w_\mathrm{aux}$. This is necessary because passing back to $q_\mathrm{B}^\mathrm{X}$  is not possible until only $\overline{\star c}$ remains from $\overline{\star}^{L+2}t_{j,r}$.
\item $Close$ the literal, i.\,e., finish consuming $\overline{\star}^{L+2}t_{j,r}$ with passing back to $q_\mathrm{B}^\mathrm{T}$.
\end{itemize}
Then, $M$ passes to $q_\mathrm{D}$ and must use the loops on $q_\mathrm{D}$ to consume a possible reminder of $u_i^{p_i}$ from $w_\mathrm{aux}$. After that, if $w_\mathrm{aux}$ is not fully consumed, the first phase must continue, i.\,e., $M$ passes back to $q_\mathrm{A}$.
\par 
Together, for each $1\le i\le n$ the automaton chooses $\mathrm{X}\in \{\mathrm{T},\mathrm{F}\}$ and then deletes from $w_\phi$ an arbitrary number of factors that stand for occurrences of the exact literal $x_i$ or $\neg x_i$, respectively.
\par
Finally, Lemma~\ref{ReducingGjfaAlph} converts $M$ to a binary GJFA $M'$ and gives the homomorphism $h$ with $\left| h(w) \right|\le 13\left|w\right|$. 
Because $\left|h(w)\right|=\landau (m\log n)$ and the construction of $w$ and $h(w)$ from $\phi$ works in linear time, any algorithm deciding whether $h(w)\in L(M')$, running in time $\landau^*\!\left(2^{o\left({\frac{\left|w\right|}{\log\left|w\right|}}\right)}\right)$, runs in time $\landau^*\!\left(2^{o(m)}\right)$ and violates ETH.
\end{proof}

A special feature of the two particular GJFAs used in the proofs of Theorems~\ref{GJFAFixedWordProbHardTheorem}~and~\ref{GJFAFixedWordProbHardThreeSatTheorem} (see Figures~\ref{fig:NPcGJFA}~and~\ref{fig:GJFA3SAT}) is that 
the length of transition labels is at most $2$ (note that for JFAs, i.\,e., machines with labels of length at most $1$, the fixed word problem lies in $\pclass$). However, Lemma~\ref{ReducingGjfaAlph}, converting the GJFAs to binary ones, increases the lengths of labels.
It is open whether the fixed word problem remains $\npclass$-hard for GJFAs with binary alphabets and with words of length at most $2$ in the transitions.\par
The hardness results presented here point out that the difference between finite machines and general finite machines is crucial if we interpret them as jumping finite automata. In contrast to this, the universal word problem for classical finite automata on the one hand and classical general finite automata on the other is very similar in terms of complexity, i.\,e., in both cases it can be solved in polynomial time.

\LV{The fact that descriptional mechanisms could yield $\npclass$-hard universal word problems if both shuffle and concatenation operations are somehow involved was already known.
For instance, in \cite{OgdRidRou78} it is remarked at the end that 
expressions formed like ordinary regular expressions, but with shuffle as an additional operator, lead to a type of expressions with 
an $\npclass$-complete universal word problem. However, unlike in the case of GJFA, the class of languages that  can be described is just $\mathscr{REG}$ and hence at least the fixed word problem lies in $\nlclass$ for this type of expressions.
}





Let us comment on one more aspect of parsing JFA. Although this mechanism was presented as a device to accept words (elements of the free monoid), the very nature of the acceptance mode brings along the idea to present input words as tuples of integers. This makes no difference as long as the integers are encoded in unary, but it does make a difference if they are encoded in binary.
This is the standard encoding for the commutative grammars / semilinear sets as studied by Huynh \cite{Huy82,Huy83}
and also explains why his complexity results seemingly deviate from ours.
More precisely, he has shown that the
universal word problem for JFA (when they are considered as processing tuples of numbers encoded in binary) is indeed $\npclass$-complete.
\subsection{Non-Disjointness and Non-Universality}

The \emph{non-disjointness problem} is the task to decide, for given automata $M_1$ and $M_2$, whether there exists a word $w$ that is accepted by both $M_1$ and $M_2$, i.\,e., whether it holds that $L(M_1) \cap L(M_2) \neq \emptyset$. In the case of JFA, we encounter a similar situation as for the universal word problem, i.\,e., it can be decided in polynomial time for fixed alphabets, while it becomes $\npclass$-complete in general~\cite{Kop2015}. \par
The $\npclass$-hardness of non-disjointness also follows easily from the proof of Theorem~\ref{JFAUniversalWordProblemETHTheorem}. Moreover, the complexity lower bound depending on ETH applies to this problem as well:
\begin{theorem}
There is no algorithm deciding, for given JFAs $M_1,M_2$ with state sets $Q_1,Q_2$, whether $L(M_1) \cap L(M_2) \neq \emptyset$ in time $\landau^*\!\left(2^{o\left(\left|Q_1\right|+\left|Q_2\right|\right)}\right)$, unless ETH fails.
\end{theorem}
\begin{proof}
The construction from the proof of Theorem~\ref{JFAUniversalWordProblemETHTheorem} produces, for given formula $\phi$ with $n$ variables and $m$ clauses, a JFA $M_1$ with $\landau(n)$ states and a word $w$ of length $m$ such that $\phi$ is satisfiable if and only if $w\in L(M_1)$. We can trivially construct a JFA $M_2$ with $\landau(m)$ states such that $L(M_2)=\perm(w)$. Then $w\in L(M_1)$ if and only if $L(M_1) \cap L(M_2) \neq \emptyset$. \par 
Thus, any algorithm answering $L(M_1) \cap L(M_2) \neq \emptyset$ in time $\landau^*\!\left(2^{o\left(\left|Q_1\right|+\left|Q_2\right|\right)}\right)$ can solve 3-SAT in time $\landau^*(2^{o(n+m)})$, which violates ETH.
\end{proof}
Another basic decision problem is the \emph{non-universality problem}, where the task is to decide, for a given automaton $M$, whether $L(M) \neq \Sigma^*$. 
Results of \cite{MeySto72} and the fact that, on unary alphabets, classical nondeterministic machines coincide with JFAs, imply that the non-universality problem for JFA is $\npclass$-hard even if restricted to JFA with unary alphabets.
On the other hand, \cite{Kop2015} shows (in terms of a more general model) that non-universality lies in $\npclass$ for any fixed alphabet size.
For the unrestricted variant of non-universality, which is trivially $\npclass$-hard as well, no close upper bound of the complexity is known \cite{Kop2015}.
{\renewcommand{\arraystretch}{1.6}
\begin{table}
\begin{centering}
\begin{tabular}{|>{\raggedright}m{25mm}|>{\centering}m{9mm}>{\raggedleft}m{4mm}|>{\centering}m{9mm}>{\raggedleft}m{4mm}|>{\centering}m{9mm}>{\raggedleft}m{4mm}|>{\centering}m{9mm}>{\raggedleft}m{4mm}|}
\cline{2-9} 
\multicolumn{1}{>{\raggedright}m{25mm}|}{} & 

\multicolumn{2}{>{\centering}m{13mm}|}{ \hspace*{5mm} \textbf{JFA} } &

\multicolumn{2}{>{\centering}m{13mm}|}{
\hspace*{4mm}\textbf{JFA} \hspace*{4mm}\mbox{\textbf{\footnotesize{$\left|\Sigma\right|\!=\! k$}}}} &

\multicolumn{2}{>{\centering}m{13mm}|}{ \hspace*{4mm}\textbf{GJFA} } &

\multicolumn{2}{>{\centering}m{13mm}|}{
\hspace*{3mm}\textbf{GJFA} \hspace*{4mm}\mbox{\textbf{\footnotesize{$\left|\Sigma\right|\!=\! k$}}}}

\tabularnewline \hline 

\textbf{fixed word~problem} & 
\hspace*{6mm}P &
\cite{JedSze2001} &
\hspace*{6mm}P &  
& 
\hspace*{4mm}NPC &  
& 
\hspace*{2mm} \begin{minipage}[t]{11mm} \begin{center} NPC{\footnotesize{ if~$k\!\ge\!2$}} \par\end{center} \end{minipage} & 
$\blacklozenge$
\tabularnewline \hline  

\textbf{universal word~problem} &
\hspace*{4mm}NPC & 
\cite{MaySto94} & 
\hspace*{6mm}P & 
\cite{JedSze2001} & 
\hspace*{4mm}NPC &  
& 
\hspace*{2mm} \begin{minipage}[t]{11mm} \begin{center} NPC{\footnotesize{ if~$k\!\ge\!2$}} \par\end{center} \end{minipage} & 
$\blacklozenge$
\tabularnewline \hline 
\textbf{non-disjointness}&
\hspace*{4mm}NPC &
\cite{Kop2015} &
\hspace*{6mm}P & 
\cite{Kop2015} &
\hspace*{4mm}UND &
& 
\hspace*{2mm} \begin{minipage}[t]{11mm} \begin{center} UND{\footnotesize{ if~$k\!\ge\!18$}} \par\end{center} \end{minipage} &
\cite{VO6}
\tabularnewline \hline 
\textbf{non-universality} & 
\hspace*{-1mm} \mbox{NP-hard} & 
\cite{MeySto72} & 
\hspace*{2mm} \begin{minipage}[t]{11mm} \begin{center} NPC{\footnotesize{ if~$k\!\ge\!1$}} \par\end{center} \end{minipage} &
\cite{Kop2015} & 
\hspace*{4mm}UND & \cite{VO9} & \hspace*{7mm}\textbf{?} & \tabularnewline
\hline  \end{tabular} \par\end{centering}
\caption{Lower and upper bounds on complexity of basic problems. Legend: NPC --- $\npclass$-complete; UND --- undecidable; \textbf{?} --- no information; $\blacklozenge$ --- present results}

\end{table}}


\LV{\section{Discussions and Prospects}

We have related the concept of jumping finite automata to the, 
actually quite well-studied, area of expressions involving shuffle operators. This immediately opens up further questions, and it also shows some limitations for this type of research programme.
\begin{itemize}
\item Is there a characterization of the class of languages accepted by general jumping finite automata in terms of expressions?\footnote{We claimed to have found such a characterization at the German Formal Language community meeting in 2014, but this claim turned out to be flawed.} We are currently working on this question and other ones related to GJFAs.
\item The original motivation for introducing variants of expressions involving shuffle operators was to model  parallel features from programming languages; see, e.\,g., \cite{CamHab74,Maz75,Sha78}. It is well-known that adding all according features immediately lead to expressions that are computationally complete, i.\,e., they characterize the recursively enumerable languages~\cite{AraTok81}.
Notably, expressions with limited nesting of iterated concatenation and iterated shuffle operators (as provided by our main normal form results for $\alpha$-SHUF expressions) have a descriptive power limited by 
Petri nets (without inhibitor arcs), so that in particular the non-emptiness problem for such limited expressions is decidable (in contrast to the general situation), confer \cite{AraKagTok81,EspNie94,Kos82,May84}.
Yet, decidability questions for Petri nets are quite hard, so that in any case the study of restricted versions of shuffle expressions or related devices is of considerable practical interest.
\item The inductive definition of $\alpha$-SHUF expressions starts with single letters (plus symbols for the empty set and the empty word). This is contrasting the definition of SHUF expressions, which starts with any finite language as a basis. As it is well-known, for classical regular expressions this difference vanishes. 
Hierarchies as the one explained in \cite{FliKud2012b} should inspire similar research for $\alpha$-SHUF expressions as introduced in this paper.
\item 
As there is a number of variations and restrictions of the shuffle operation itself~\cite{KarSos2005,Kud97,KudMat97,MatRozSal98}, it would be also interesting to study
expressions that contain some of these. We plan to deal with this topic in the near future.
\item The whole area seems to be related to \emph{membrane systems}, also known as \emph{P systems}. The reason is that membrane computing often reduces to \emph{multiset computing}, which is just another name for dealing with subsets of $\mathbb{N}^\Sigma$. These connections are explained by Kudlek and Mitrana in \cite{KudMit0203}.
\item We have somehow initiated the study of complexity aspects of JFA and related models under ETH. Many other automata problems can be investigated in this paradigm (as also indicated in the Appendix), and more importantly from an algorithmic point of view, it would be interesting to know of procedures that match the proven lower bounds.
\end{itemize}
Summarizing, the study of expressions involving the shuffle operation,
as well as of variants of jumping automata, 
still offers a lot of interesting questions, as it is also
indicated in the recent survey of Restivo~\cite{Res2015}.
}

%
\section*{Acknowledgements}
\noindent
Meenakshi Paramasivan --- Supported by the DAAD Phd Funding Programme - Research Grants for Doctoral Candidates and Young Academics and Scientists - Programme ID: 57076385. \\
Vojt\v{e}ch Vorel --- Supported by the Czech Science Foundation grant GA14-10799S and the GAUK grant No. 52215.

\section*{References} 

\bibliographystyle{plain}
\bibliography{abbrev,hen,zukopieren,vojtaTemp}

\section{Appendix: The Exponential Time Hypothesis\label{Sec:ETHappendix}}

The Exponential Time Hypothesis (ETH) was formulated by Impagliazzo, Paturi and Zane
in \cite{ImpPatZan2001}: 
\begin{quote}
There is a positive real
$s$ such that 3SAT with  $n$ variables and $m$ clauses cannot be solved in time $2^{sn}(n + m)^{\landau(1)}$.
\end{quote}

ETH considerably strengthens the well-known and broadly accepted hypothesis
that  $\pclass\neq \npclass$.
A slightly weaker but more compact formulation of ETH  (which we will hence adopt
in this paper) is the following hypothesis:

\begin{quote}There is no algorithm that solves 3SAT with  $n$ variables and $m$ clauses  in time $\landau^*(2^{o(n)})$.
\end{quote}
The famous \emph{sparsification lemma}
of Impagliazzo, Paturi and Zane~\cite{ImpPatZan2001}
can hence be stated as follows:

\begin{quote}Assuming ETH, there is no algorithm that solves 3SAT with  $n$ variables and $m$ clauses  in time $\landau^*(2^{o(n+m)})$.
\end{quote}

Notice that this seemingly minor modification gives in fact a tremendous advantage to everybody aiming at proving complexity statements that hold, unless ETH fails.
Observe that, in order to prove such complexity results, one also needs a special type of reductions called SERF reductions, see \cite{ImpPatZan2001},
or also the survey \cite{LokMarSau2011b}. This type of reduction is essentially 
a subexponential-time Turing reduction, and the very power of such a a type of reduction is exploited in the sparsification lemma.
However, the easiest case of such a reduction is in fact a 
polynomial-time many-one reduction from 3-SAT such that there is a linear dependence of the size measure of the reduced instance on the number of variables and clauses of the given 3SAT instance. Several (but not all)
textbook reductions enjoy this kind of \emph{linearity property}.

It is worth mentioning that not all textbook reductions enjoy this linearity property. For instance, out of the five problems reduced (directly or indirectly) from 3-SAT in Sec. 3.1 of \cite{GarJoh79}, only \textsc{Vertex Cover} and \textsc{Clique} enjoy this property. Conversely, the given reduction from 3-SAT to \textsc{3-Dimensional Matching} (3-DM) produces $\landau(n^2m^2)$ many triples from a given 3-SAT instance with  $m$ clauses and $n$ variables.
Hence, under ETH we can only rule out algorithms running in time $\landau^*(2^{o(\sqrt[4]{t})})$ for \textsc{3-Dimensional Matching} instances with $t$ triples.
So, we might need new (and possibly also more complicated) reductions to make proper use of ETH. This venue is also exemplified by reductions presented in this paper (see Theorems \ref{JFAUniversalWordProblemETHTheorem} and \ref{GJFAFixedWordProbHardThreeSatTheorem}).

Notice that the lower bounds that can be obtained by using published proofs that only go back to 3-SAT by a chain of reductions could be really weak.
We make this statement clearer by one concrete example.
In the following, we denote the `loss' that is incurred by a reduction by given the `root term'; for instance, we have the following losses:
\begin{itemize}
\item From 3-SAT to 3-DM: $\sqrt[4]{\cdot}$ \cite{GarJoh79}.
\item From 3-DM to 4-\textsc{Packing}:  $\sqrt[4]{\cdot}$ \cite{GarJoh79}.
\item From 4-\textsc{Packing} to 3-\textsc{Packing}: $\sqrt[2]{\cdot}$ \cite{GarJoh79}.
\item From 3-\textsc{Packing} to \textsc{Exact Block Cover} EBC$_2$: linear \cite{JiaSXXZZ2014}.
\item From EBC$_2$ to the fixed word problem of a GJFA: linear (Theorem \ref{JFAUniversalWordProblemETHTheorem}).
\end{itemize}

This chain of reduction would hence incur a loss of $\sqrt[32]{\cdot}$, which also 
shows the need to exhibit yet another reduction for the purpose of making proper use of ETH (Theorem \ref{GJFAFixedWordProbHardThreeSatTheorem}).

Another example for a reduction from 3-SAT that only gives a weak-looking bound is offered by the reduction of Stockmeyer and Meyer (mentioned several time throughout this paper) that shows that the question whether any word is not accepted by a given unary regular language is $\npclass$-hard. As can be seen by analyzing that proof, 
under ETH only only the existence of 
an $\landau^*(2^{\sqrt[4]{q-\varepsilon}})$-time algorithm is ruled out (for $q$-state unary NFA).
As presented in
a talk on
\emph{Lower Bound Results for Hard Problems Related to Finite Automata}
in the workshop \emph{Satisfiability Lower Bounds and Tight Results for Parameterized and Exponential-Time Algorithms} at the Simons Institute, Berkeley, in early November 2015, this 
can be improved to the following statement:

\begin{theorem}\label{thm-tally-NFA-universality} 
Unless ETH fails, there is no $\landau^*(2^{o(q^{1/3})})$-time algorithm for deciding,
given a unary NFA $M$ on $q$ states, whether $L(M)\neq\{a\}^*$. 
\end{theorem}

As this seems to be the currently best bound of its kind, we make use of it in Corollaries \ref{Cor:regularityETH} and \ref{Cor:commutativityETH}.
Stronger reductions are known for the case of binary input alphabets, as shown in the same talk.

\end{document}